\documentclass[11pt,a4paper,english]{article}

\usepackage{array}
\usepackage{amsmath}
\usepackage{amscd}
\usepackage{latexsym}
\usepackage[english]{babel}
\usepackage{amsfonts}
\usepackage{version}
\usepackage{graphicx}
\usepackage{changepage}
\usepackage[T1]{fontenc}
\usepackage[utf8]{inputenc}
\usepackage{amsthm}
\usepackage{mathtools}
\usepackage{amssymb}
\usepackage{version}
\usepackage{fancyhdr}
\usepackage{caption}
\usepackage{stmaryrd}
\usepackage{url}
\usepackage{hyperref}
\usepackage{graphicx}
\usepackage{subfigure}
\usepackage{color,soul}

\newcommand{\esp}[2][\mathbb E] {#1\left[#2\right]}

\numberwithin{equation}{section}

\newcommand{\condespg}[2][\G_t]       { E\left.\left[#2\right|#1\right]}

\newcommand{\ud}{\mathrm d}

\newcommand{\R}{\R}

\newtheorem{theorem}{Theorem}[section]
\newtheorem{lemma}[theorem]{Lemma}
\newtheorem{defn}[theorem]{Definition}
\newtheorem{assump}[theorem]{Assumption}
\newtheorem{prop}[theorem]{Proposition}
\newtheorem{cor}[theorem]{Corollary}
\newtheorem{rem}[theorem]{Remark}

\def \H {{\mathcal H}}
\def \Hb {{\mathbb H}}
\def \G {{\mathcal G}}
\def \Gb {{\mathbb G}}
\def \F {{\mathcal F}}
\def \Fb {{\mathbb F}}
\def \P {{\mathbb P}}

\def \R {{\mathbb R}}
\def \N {{\mathbb N}}

\newcommand{\Ind}[1]{\mathbf{1}_{\left\{ #1 \right\}}}

\def \P {\mathcal{P}}

\newcommand{\mail}[1]{\href{mailto:#1}{\texttt{#1}}}

\title{Extended Reduced-Form Framework for Non-Life Insurance}
\author{Francesca Biagini\footnote{Main affiliation: Department of Mathematics, LMU Munich, Theresienstra{\ss}e, 39, 80333 Munich, Germany, Email: \mail{biagini@math.lmu.de}} \footnote{Secondary affiliation: Department of Mathematics, University of Oslo, Box 1053, Blindern, 0316, Oslo, Norway.}
\and Yinglin Zhang\footnote{Department of Mathematics, LMU Munich, Theresienstra{\ss}e, 39, 80333 Munich, Germany, Email: \mail{zhang@math.lmu.de}}}

\begin{document}
\maketitle 

\begin{abstract}
	In this paper we propose a general framework for modeling an insurance liability cash flow in continuous time, by generalizing the reduced-form framework	for credit risk and life insurance. In particular, we assume a nontrivial
	dependence structure between the reference filtration and the insurance internal filtration. We apply these results for pricing and hedging non-life insurance liabilities in hybrid financial and insurance markets, while taking into account the role of inflation under the benchmarked risk-minimization approach. This framework offers at the same time a general and flexible structure, and explicit and treatable pricing-hedging formula.
\end{abstract}

\section{Introduction}

In this paper we propose a general framework for modeling an insurance liability cash flow in continuous time, by extending the classic reduced-form setting for credit risk and life insurance. In particular, we consider a nontrivial dependence structure between the reference filtration $\Fb$  and insurance internal filtration $\Hb$.
The global information flow available to the insurance company is represented by $\Gb = \Fb \vee \Hb$.
In this way, we obtain for the first time a framework in continuous time for non-life insurance, where filtration dependence is taken into account.
In view of the development of insurance-linked derivatives, which offer the possibility of transferring insurance risks to the financial market, this bottom-up modeling approach can be used for pricing and hedging both life and non-life insurance liabilities in hybrid financial and insurance markets.
As an application of the general framework structure, we derive pricing and hedging formulas for non-life insurance claims by taking into account the role of inflation under the benchmarked risk minimization approach, as introduced in \cite{Pla3}, \cite{Bia-Cre} and \cite{Bia-Zha}.

Historically, the mathematical modeling of life and non-life insurance liabilities in continuous time is quite asymmetric and risk mitigation of non-life portfolios via asset allocation is scarcely practiced. While there are a lot of recent works concerning life insurance, see e.g \cite{Mol-u}, \cite{Cai-Bla}, \cite{Dah-Mol}, \cite{Bar}, \cite{Bia-Rhe-h}, \cite{Bia-Sch}, \cite{Bia-Rhe-r}, \cite{Bia-Zha}, non-life insurance is often studied in discrete time and/or state space, see e.g. \cite{Kal-Nor}, \cite{Lar}, \cite{Hap-Mer}. We refer to e.g. \cite{Wut} for a unified framework for life and non-life insurance in discrete time.
Mathematical frameworks for non-life insurance in continuous time can be found in e.g. \cite{Mol-i}, \cite{Del}, \cite{Bar-Lau}, \cite{Nor-Sav}, \cite{Nor-qua} and \cite{Schm}. However, these settings do not consider a nontrivial dependence structure between reference filtration and insurance internal filtration. In particular, in e.g. \cite{Nor-Sav} and \cite{Nor-qua}, the insurance internal filtration is not distinguished from the reference filtration, and in e.g. \cite{Mol-i}, \cite{Del} and \cite{Bar-Lau}, reference and insurance internal filtrations are assumed to be independent. 
The importance of considering a nontrivial dependence structure between filtrations, which represent different information flows in a hybrid market, is discussed in \cite{Bia} in view of the recent introduction of insurance-linked derivatives.
Derivatives based on occurrence intensity index, such as mortality derivatives, weather derivatives etc., can play an important role in mitigating risks of insurance companies in the case of life and non-life insurance business. In particular, non-catastrophe non-life insurance (see e.g. \cite{Bri} for the distinction between catastrophe and non-catastrophe insurance.), which includes car insurance, theft insurance, home insurance, etc., as opposed to catastrophe non-life insurance, covers high-probability low-cost events, and is often neglected by the literature. 
This paper aims to fill this gap, i.e. to identify the role of  non-life related hybrid products, to contribute to their design and to provide analytical results which can be used for the non-life insurance reserving problem and the valuation and hedging of non-life insurance portfolio by trading hybrid non-life products, which are currently still not common but can be potentially attractive in the future.

Recent non-life insurance literature in continuous time, see e.g. \cite{Bar-Lau}, \cite{Nor-Sav}, \cite{Nor-qua} and \cite{Schm}, commonly assumes the insurance internal information flow as given by the natural filtration of a marked point process, which describes the insurance claim movement.
Pricing and hedging formulas are then obtained by using the compensator of this marked point process.
However, as we discuss in Section \ref{sec: relation with compensator}, this approach can not be always followed in the case of multiple filtrations with nontrivial dependence. Indeed, with respect to a generic filtration, it is not always true that there exists a marked point process with a given compensator, and the compensator does not always determine uniquely the law of the process. 
To overcome these difficulties, we propose a new framework, which uses a direct approach as in Section 5.1 and 9.1.2 of \cite{Bie-Rut} and allows an explicit bottom-up construction to treat more general filtrations. 
We note that, when our general framework is reduced to the case of life insurance, the compensator approach and the direct modeling approach coincide, see the discussion in e.g. \cite{Bie-Rut} for the classic reduced-form framework.

More precisely, in our new framework we consider a homogeneous insurance portfolio with $n$ claims. We assume that the reference filtration $\Fb$ includes information related to the financial market and to environmental, social and economic indicators.
Following the classic non-life insurance modeling approach as in e.g. \cite{Ary} and \cite{Mik}, we assume  the insurance internal filtration $\Hb$, which represents internal information of an insurance company given by the claim movements, to be generated by a family of marked point processes, describing sequences of reporting times and associated losses.
As typically in the case of non-life insurance, accident times and their related damages are unknown until the moment of reporting.
We are able to capture these features and at the same time to introduce a dependence structure between filtrations $\Fb$ and $\Hb$ by providing a nontrivial extension of the classical reduced-form framework. In particular, we model accident times as $\Fb$-conditionally independent random variables with a common $\Fb$-adapted intensity process $\mu$. We note that in this way, we assume that intensity of accident times $\mu$ may be influenced by external factors and economic indicators. Random delay between accident time and the first reporting is modelled in the first mark and subsequent development of the claim is modelled by a time shift of an independent marked point process with respect to the first reporting.
This structure includes the life insurance case and allows to obtain analytical valuation formulas, which can be expressed in term of the accident intensity $\mu$, the delay distribution and the updating distribution, as illustrated in the preliminary calculations in Section \ref{sec: preliminary results}.
We then apply these results for pricing and hedging insurance liabilities in a hybrid market under the benchmarked risk-minimization approach of \cite{Pla3}, \cite{Bia-Cre} and \cite{Bia-Zha}. The hybrid nature of the combined market is given by the (hypothetical) presence of derivatives related to the intensity process $\mu$ on the financial market and by the influence of inflation and benchmark portfolio in the valuation and hedging of insurance liabilities. This gives insight to the role of non-life insurance linked financial products, which might be attractive for non-life insurance companies to mitigate and transfer non-life underwriting risks, in a similar manner to the role of mortality or longevity derivatives for life insurance companies.

In summary, the main contributions of this paper is to propose a general framework for modeling an insurance liability cash flow  in continuous time taking into account filtration dependence, which includes both  the non-life and the life insurance case. This approach allows to obtain liability cash flow valuation formulas for non-life insurance portfolios in a general setting, as well as  for the pricing of (possible) hybrid financial products written as a bet on non-life insurance events or portfolios. Whereas these hybrid products are uncommon in the present market,  our framework could help to identify their role and contribute to their design and valuation.

This paper is organised as follows. In Section \ref{sec: general framework} we construct a general framework for insurance liability cash flow in continuous time under a nontrivial dependence between the reference and the insurance internal filtrations, applicable both to life and non-life insurance, and give a brief comparison with the existing insurance frameworks in the literature. In Section \ref{sec: preliminary results} we give some useful preliminary valuation results in this setting. In Section \ref{sec: relation with compensator} we discuss the compensator approach.
In Section \ref{sec: hybrid market} we describe the hybrid nature of the combined market and derive the real world pricing formula and benchmarked risk-minimizing strategy for non-life insurance claims. 

\section{General framework}\label{sec: general framework}

In this section we construct a general framework for modeling an insurance liability cash flow. We consider a filtered probability space $(\Omega, \G, \Gb, P)$, where $\Gb := (\G_t)_{t \geqslant 0}$, $\G = \G_\infty$, and $\G_0$ is trivial. 

We assume that $\Gb =\Fb \vee \Hb$, where $\Fb := (\F_t)_{t \geqslant 0}$ and $\Hb := (\H_t)_{t \geqslant 0}$ are filtrations representing respectively a reference information flow and the internal information flow only available to the insurance company. Hence $\Gb$ describes the global information flow available to the insurance company. The reference filtration $\Fb$ typically includes information related to the financial market, and to
environmental, political and social indicators. While we do not specify the structure of the reference filtration $\Fb$, we assume that the insurance internal filtration $\Hb$ is generated by a family of marked point processes, representing the times and amounts of losses of the insurance portfolio, as in e.g. \cite{Ary}, \cite{Jew}, \cite{Nor-pre} and \cite{Nor-preII}.
Filtrations $\Fb$ and $\Hb$ are not supposed be independent. Without loss of generality, we assume that all filtrations satisfy the conditions of completeness and right-continuity. If not otherwise specified, all relations in this paper hold in the $P$-a.s. sense.
For a detailed background of marked point processes we refer to e.g. \cite{Las}, \cite{Dal} and \cite{Jac}. In the following we use the classic terminology of non-life insurance, see e.g. \cite{Wut-Mer-s} and \cite{Par}, and specify the filtration $\Hb$ as follows.

We consider an insurance portfolio with $n$ policies. 
For $i$-th policy with $i = 1,...,n$, the insurance company is typically informed about the accident occurred at a random time $\tau^i_0$ only after a random delay $\theta^i$, which can be very long especially in the case of non-life insurance.
Once the accident is reported at time $\tau^i_1$, where 
\begin{equation}\label{eq: relation tau1 tau0}
	\tau^i_1 := \tau^i_0  + \theta^i,
\end{equation}
both the accident time $\tau^i_0$, the reporting delay $\theta^i$ and the impact size of the accident, described by a nonnegative random variable $X^i_1$, become available information. In particular, we assume that for all $i=1,...,n$, $\tau^i_0 > 0$ $P$-a.s.

Let $\N_0$ be the set of natural numbers without zero.
We describe the $i$-th insurance policy movement by a marked point process $(\tau^i_j, \Theta^i_j)_{j \in \N_0}$
with 2-dimensional nonnegative marks. That is, the sequence $({\tau}^i_j)_{j \in \N_0}$ is a point process, where 
\[
	{\tau}^i_j: (\Omega, \G, P) \rightarrow (\overline{\R}_+, \mathcal{B}(\overline{\R}_+)), \ \ \ \ j \in \N_0,
\]
and $({\Theta}^i_j)_{j \in \N_0}$ is a sequence of 2-dimensional nonnegative random variables, with
\[
	{\Theta}^i_j: (\Omega, \G, P) \rightarrow ({\R}^2_+, \mathcal{B}({\R}^2_+)), \ \ \ \ j \in \N_0.
\]
For every $j \in \N_0$, the random time $\tau^i_j$ describes the reporting time of $j$-th event related to $i$-th policy. The mark components $\Theta^i_j$ describe the reporting delay and the impact size of the corresponding event, respectively, which are known only if the event is reported. 
More precisely, we set
\begin{equation}\label{eq: specif MP 1}
	\tau^i_1 \ \ \ \ \text{with mark} \ \ \ \  \Theta^i_1 = (\theta^i, X^i_1),
\end{equation}
and
\begin{equation}\label{eq: specif MP 2}
	\tau^i_{j + 1} = \tau^i_1 + \tilde{\tau}^i_j \ \ \ \ \text{with mark} \ \ \ \ \Theta^i_{j+1} = (0, {X}^i_{j + 1}) := (0, \tilde{X}^i_j),
\end{equation}
for $j \geqslant 1$, where $(\tilde{\tau}^i_j, \tilde{X}^i_{j})_{j \in \N_0}$ is an auxiliary marked point process, which describes updating and development after the first reporting at ${\tau}^i_1$. 
For the sake of simplicity, we here assume that only the first reporting delay is different from zero, since in this paper we focus on modeling the first accident times $\tau^i_0$ and their relation with the reference filtration. However  our setting can be easily generalized by considering non zero random delays in (\ref{eq: specif MP 2}). 
We assume that the marked point process $(\tilde{\tau}^i_j, \tilde{X}^i_{j})_{j \in \N_0}$ is simple, i.e.
\[
	\lim_{j \rightarrow \infty}{\tilde{\tau}}^i_j = \infty,
\]
and ${\tilde{\tau}}^i_j < {\tilde{\tau}}^i_{j+1}$, if ${\tilde{\tau}}^i_j < \infty$, and satisfies the following integrability condition
\begin{equation}\label{eq: finite moment}
	E\left[\sum_{j=1}^\infty \Ind{\tilde{\tau}^i_j \leqslant t} \tilde{X}^i_{j} \right] < \infty \ \ \ \ \text{for all } t \geqslant 0,
\end{equation}
for $i=1,...,n$.
In particular, the random times $(\tau^i_j)_{j \in \N_0}$ are strictly ordered in case of finite value: 
\begin{equation}\label{eq: simple}
	\begin{array}{l}
	\tau_1^i < \tau_2^i < \cdots < \tau_j^i < \tau_{j+1}^i < \cdots, \ \ \ i=1,...,n.
	\end{array}
\end{equation}
Note that we may have $\infty = \tau^i_j = \tau^i_{j+1} = ...$, in such a case infinite value stands for an event which never happens. 
For the sake of simplicity we assume also the following.
\begin{assump}\label{ass: homo and indep} \ 
	\begin{enumerate}
		\item \emph{Homogeneous delay}: the random delays $\theta^i$, $i = 1,...,n$, have the same distribution. 
		
		\item \emph{Homogeneous development}: the marked point processes $(\tilde{\tau}^i_j, \tilde{X}^i_{j})_{j \in \N_0}$, $i = 1,...,n$, have the same distribution.
		
		\item \emph{Independent first mark}: the first marks $X^i_1$, $i = 1,...,n$, are mutually independent and independent of $\F_\infty \vee \sigma(\tau^1_0) \vee ... \vee \sigma(\tau^n_0)$.
		
		\item \emph{Independent delay}: the random delays $\theta^i$, $i=1,...,n$, are mutually independent and independent of $\F_\infty \vee \sigma((\tau^1_0, X^1_1)) \vee ... \vee \sigma((\tau^n_0, X^n_1))$. 
		
		\item \emph{Independent development}: the marked point processes $(\tilde{\tau}^i_j, \tilde{X}^i_{j})_{j \in \N_0}$, $i=1,..,n$ are mutually independent and independent of $\F_\infty \vee \sigma((\tau^1_1, \theta^1, X^1_1)) \vee ... \vee \sigma((\tau^n_1, \theta^n, X^n_1))$. 
	\end{enumerate}
\end{assump}

\noindent We emphasize that the above assumptions cover sufficient generality. The homogeneity assumptions can be satisfied by subdividing appropriately the insurance portfolio in homogeneous groups of claims. The independence assumptions reflect the fact that reporting delays $\theta^i$, occurrences and size of the losses after the first reporting time, described by $(\tilde{\tau}^i_j, \tilde{X}^i_j)_{j \in \N_0}$, are typically idiosyncratic factors which are independent of each other and independent of the reference information. However, we introduce a dependence structure by modeling $\Fb$-progressively measurable occurrence intensities of the accidents, as we will present in (\ref{eq: def mu}) and (\ref{eq: def gamma}). This will reflect the assumption that the occurrence intensity of accidents can be deduced from the reference information flow represented by $\Fb$, while further updates of accident events $(\tilde{\tau}^i_j, \tilde{X}^i_j)_{j \in \N_0}$ are typically insurance portfolio specific and are not available as third-party or reference information.
We assume furthermore that the distribution of delay variables $\theta^i$, $i=1,...,n$ has the following structure.

\begin{assump}\label{ass: delay}
	The common cumulative distribution function $G: [0,+\infty) \rightarrow [0,1]$ of $\theta^i$, $i=1,...,n$, assigns probability $\alpha_0$ at 0 and has a density function $g$ for $x>0$, i.e,
	\begin{equation}\label{eq: delay structure}
		G(x) = \alpha_0 + \int_0^{x} g(y) \ud y, \ \ \ \ x \in \R_+.
	\end{equation}
\end{assump}

\noindent According to the above assumption, the delays are nonnegative and may have a mixed distribution. In this way, we cover both the case of life insurance with $\theta^i \equiv 0$, i.e. $g = 0$, and the case of non-life insurance with non-null delays. 

For every $i= 1, ..., n$, we define the marked cumulative process $N^i$ by
\[
	N^i(t,B) (\omega) :=  \sum_{j=1}^\infty \Ind{\tau^i_j(\omega) \leqslant t} \Ind{\Theta^i_j(\omega) \in B}, \
\]
for every $t \geqslant 0$, $B \in \mathcal{B}(\R^2_+)$, $\omega \in \Omega$.
The process $(\mathbf{N}^i_t)_{t \geqslant 0}$ defined by
\[
	{\mathbf{N}^i_t}:= N^i(t,\R^2_+) = \sum_{j=1}^\infty \Ind{\tau^i_j \leqslant t}, \ \ \ \ t \geqslant 0,
\]
is called ground process associated to the marked point process.
At any time $t \geqslant 0$, the random variable ${\mathbf{N}^i_t}$ counts the number of occurrence of $\tau^i_j$ up to time $t$.
In the literature, the name marked point process refers sometimes to the process $N^i$. 
Indeed, there is a unique correspondence between the marked point process $(\tau^i_j, \Theta^i_j)_{j \in \N_0}$ and its marked cumulative process $N^i$. More precisely, $\tau^i_j$ can be uniquely defined by $N^i$ as
\begin{equation}\label{eq: rel mpp and mcp 1}
	\{ \tau^i_j \leqslant t \} = \{ {\mathbf{N}^i_t} \geqslant j \},
\end{equation}
for all $t \geqslant 0$, and $\Theta^i_j$ can be uniquely defined by $N^i$ as
\begin{equation}\label{eq: rel mpp and mcp 2}
	\{\Theta^i_j \in B \} = \bigcup_{K'=1}^\infty \bigcap_{K=K'}^\infty \bigcup_{k=1}^\infty \{{\mathbf{N}^i_{(k-1)/2^K}} = n-1, N^i(k/2^K, B) - N^i((k-1)/2^K, B) = 1\},
\end{equation}
for all $B \in \mathcal{B}(\R^2_+)$. See equations (2.8), (2.9) of \cite{Jac} and Lemma 2.2.2 of \cite{Las}.

We consider the filtrations ${\Hb}^{i,1} := ({\H}^{i,1}_t)_{t \geqslant 0}$ with 
\[
	{\H}^{i,1}_t := \sigma \left( \Ind{\tau^i_1 \leqslant s} \Ind{(\theta^i, X^i_1) \in B}, 0 \leqslant s \leqslant t, \text{ for all } B \in \mathcal{B}(\R^2_+) \right),
\]
for all $t \geqslant 0$, and ${\Hb}^{i,j} := ({\H}^{i,j}_t)_{t \geqslant 0}$, $j >1$, with 
\[
	{\H}^{i,j}_t := \sigma \left( \Ind{\tau^i_j \leqslant s} \Ind{X^i_j \in B}, 0 \leqslant s \leqslant t, \text{ for all } B \in \mathcal{B}(\R_+) \right),
\]
for all $t \geqslant 0$. It holds that
\[
	{\H}^{i,j}_\infty = \sigma(\tau^i_j) \vee \sigma( X^i_j) \ \ \ \ \text{for } j > 1.
\]
In particular, in view of (\ref{eq: relation tau1 tau0}) we have
\begin{equation} \label{eq: H 1 infty}
	{\H}^{i,1}_\infty = \sigma(\tau^i_1) \vee \sigma((\theta^i, X^i_1)) = \sigma(\tau^i_0) \vee \sigma((\theta^i, X^i_1)).
\end{equation}
Let $\Hb^i := (\H^i_t)_{t \geqslant 0}$ be the natural filtration of the marked cumulative process $N^i$, that is for all $t \geqslant 0$,
\begin{equation*}
	\H^i_t = \sigma (N^i(s, B), 0 \leqslant s \leqslant t, \text{ for all } B \in \mathcal{B}(\R^2_+)).
\end{equation*}
The internal information flow of the insurance company is described by the filtration $\Hb := (\H_t)_{t \geqslant 0}$, where 
\begin{equation}\label{eq: def H}
	\H_t := \H^1_t \vee ... \vee \H^n_t, \ \ \ \ t \geqslant 0.
\end{equation}

Similarly, for $i=1,..,n$, we call $\tilde{N}^i$ the corresponding marked cumulative processes associated to the marked point processes $(\tilde{\tau}^i_j, \tilde{X}^i_{j})_{j \in \N_0}$ and $\tilde{\Hb}^i$ the corresponding filtration, respectively. Similarly, all other notations associated to these last processes will be denoted by the symbol "$\sim$".

\begin{lemma}\label{lemma: representation Hi}
	For every $i = 1,...,n$, we have $\Hb^i = \bigvee_{j \in \N_0} \Hb^{i,j}$.
\end{lemma}

\begin{proof}
	Clearly, we have
	\[
		\H^i_t \subseteq \bigvee_{j \in \N_0} \H^{i,j}_t.
	\]
	For the other inclusion, it is sufficient to show that for all $0 \leqslant s \leqslant t$ and $B \in \mathcal{B}(\R^2_+)$,
	\[
		\{\tau^i_j \leqslant s\} \cap \{\Theta^i_j \in B\} \in \H^i_t.
	\]
	This follows directly from (\ref{eq: rel mpp and mcp 1}) and (\ref{eq: rel mpp and mcp 2}).
\end{proof}

\noindent We now introduce the following notation, which is useful in the sequel.
For $i = 1,...,n$, $j \in \N_0$, we define 
\[
	\Hb^{i,\leqslant j} := \bigvee_{k \leqslant j} \Hb^{i,k}, \ \ \ \ \Hb^{i,\geqslant j} := \bigvee_{k \geqslant j} \Hb^{i,k},
\]
similarly for $\Hb^{i,>j}$ and $\Hb^{i,< j}$. In particular, in the case of $j = 1$, we set $\H^{i,< 1}_t := \{ \emptyset, \Omega \}$ for every $t \geqslant 0$. The following corollary is a direct consequence of Lemma \ref{lemma: representation Hi}.

\begin{cor}\label{cor: representation Hi}
	For every $i = 1,...,n$, $j \in \N_0$, we have 
	\[
		\Hb^i = \Hb^{i,\leqslant j} \vee \Hb^{i,> j} = \Hb^{i,< j} \vee \Hb^{i,\geqslant j}.
	\]
\end{cor}

Similarly to the reduced form setting for credit risk and life insurance, we now model the accident times $\tau^i_0$, $i = 1,...,n$, and their relation with the reference filtration in the following way. 
We assume that random times $(\tau^i_j)_{j \in \N}$, $i = 1,...,n$ are not $\Fb$-stopping times, i.e. for every $i$ and $j$, there is $t$ such that $\{\tau^i_j \leqslant t\} \notin \F_t$. As in Section 9.1.2 of \cite{Bie-Rut}, we set that accident times $\tau^i_0$, $i = 1,...,n$, are such that for $t \in [0,\infty]$ and $s \in [0,t] \cap [0,\infty)$, 
\begin{equation}\label{eq: cond exp tau 0}
	P\left. \left(\tau^i_0 > s \right| \F_t \right) = P\left. \left(\tau^i_0 > s \right| \F_s \right),
\end{equation}
and for  $l, k = 1,..,n$ with $l \neq k$, $\tau^l_0$ and $\tau^k_0$ are $\Fb$-conditionally independent, i.e. if $t \in [0,\infty]$ and $r,s \in [0,t] \cap [0,\infty)$, we have 
\begin{equation}\label{eq: F independent tau 0}
	P\left. \left(\tau^l_0 > r, \tau^k_0 > s \right| \F_t \right) = P\left. \left(\tau^l_0 > r \right| \F_t \right) P \left. \left(\tau^k_0 > s \right| \F_t\right).
\end{equation}

\begin{rem}\label{rem: H0 are F independent}
If we define $\H^{i,0}_t := \sigma \left( \Ind{\tau^i_0 \leqslant s}: 0 \leqslant s \leqslant t \right)$, $i = 1,...,n$, then condition (\ref{eq: cond exp tau 0}) is equivalent to 
\[
	E[ X | \F_t] = E[ X | \F_s],
\]
for each integrable $\H^{i,0}_s$-measurable random variable $X$.
Condition (\ref{eq: F independent tau 0}) is equivalent to the $\F_t$-conditional independence between the $\sigma$-algebras $\H^{l,0}_t$ and $\H^{k,0}_t$.
\end{rem}

\noindent Furthermore, if $F^i := (F^i_t)_{t \geqslant 0}$ is the $\Fb$-conditional cumulative process  of $\tau^i_0$,
\[
	F^i_t := P\left.\left( \tau^i_0 \leqslant t \right|  \mathcal{F}_t\right), \ \ \ \ t \geqslant 0,
\]
we assume that there exists a locally integrable and $\Fb$-progressively measurable process $\mu^i := (\mu^i_t)_{t \geqslant 0}$, such that
\begin{equation}\label{eq: def mu}
	e^{- \int_0^t \mu^i_s \ud s} = 1 - F^i_t \ \ \ \ \text{for all } t \geqslant 0.
\end{equation}
We define $\Gamma^i := (\Gamma^i_t)_{t \geqslant 0}$ as
\begin{equation}\label{eq: def gamma}
	\Gamma^i_t := \int_0^t \mu^i_s \ud s, \ \ \ \  t \geqslant 0.
\end{equation}
The process $\mu^i$ is called \emph{intensity process} of the random jump time $\tau^i_0$ and the process $\Gamma^i$ is called \emph{hazard process} of $\tau^i_0$. 
An explicit construction in Example 9.1.5 of \cite{Bie-Rut} shows that for a given family of locally integrable $\Fb$-progressively measurable process $\mu^i$, $i = 1,...,n$, it is always possible to construct random times $\tau^i_0$, $i = 1,...,n$, such that $\Gamma^i$ is the hazard process of $\tau^i_0$ for every $i = 0,...,n$, and all the assumptions above are satisfied. 
For the sake of simplicity, we assume that the insurance portfolio is homogeneous.
\begin{assump} \label{ass: accident times}
	The accident times $\tau^i_0$, $i= 1,...,n$, have the same intensity process $\mu$.
\end{assump}

\noindent Under this homogeneity condition, we denote the common $\Fb$-conditional cumulative function and hazard process respectively by $F$ and $\Gamma$. The above assumption reflects the fact that, while the policy developments may not have direct link to the information flow $\Fb$, the accident occurrences $\tau^i_0$, $i = 1,...,n$, are influenced by some common 
external systematic risk factors described by occurrence intensity $\mu$ and observable from filtration $\Fb$.

We now show how the general framework described above comprehends in a synthetic way both life and non-life insurance modeling, and compare our setting with the existing literature.

\subsection{Life insurance}

Life insurance policies typically do not have reporting delay and depend only on $\tau^i_0$, $i = 1,...,n$, which actually represent the decease times.
This can be included in our framework by setting $\theta^i \equiv 0$, $\tau^i_j \equiv \infty$ for all $j > 1$ and $X^i_j \equiv 1$ for all $j \in \N_0$, and interpreting $\tau^i_0$ as the decease time of person $i$, where $i = 1,...,n$. The filtration $\Gb$ is hence reduced to 
\[
	\Gb = \Fb \vee \Hb^1 \vee ... \vee \Hb^n,
\]	
where
\[
	\H^i_t = \sigma \left( \Ind{\tau^i_0 \leqslant s}, 0 \leqslant s \leqslant t \right), \ \ \ \ t \geqslant 0, \ i = 1,...,n.
\]
In particular, the $\Fb$-progressively measurable process $\mu$ is interpreted as mortality intensity in this context. 
The financial market is typically assumed to include mortality or longevity linked derivatives, such as longevity bond, which pays off the longevity index value $e^{- \int_0^T \mu_s \ud s}$ at maturity $T$.

Life insurance within hybrid market under this setting has been intensively studied in the literature, see e.g. \cite{Bar}, \cite{Bia-Sch}, \cite{Bia-Rhe-r} and \cite{Bia-Zha}. 

\subsection{Non-life insurance}\label{sec: non life}

The framework in Section \ref{sec: general framework} in its full generality describes the case of non-life insurance. Indeed, non-life insurance policies typically have reporting delay, i.e. $\theta^i \neq 0$, which can also count to several years. For $i$-th policy, we interpret $X^i_j$ as payment amount at the $j$-th random times $\tau^i_j$; the exact accident time $\tau^i_0$ and first payment amount $X^i_1$ is known only after reporting at time $\tau^i_1$. Further developments  may occur after the first reporting and before the settlement of claim. The total number of developments $(\tau^i_j)_{j \in \N_0}$ is unknown as well as the amount of corresponding payments $(X^i_j)_{j \in \N_0}$. The accident time $\tau^i_0$ admits an $\Fb$-progressively measurable intensity process $\mu$ related to the underlying risk. If liquidly traded derivatives related to the $\mu$ process are available on the financial market, they could be used for hedging systematic risks related to non-life portfolio.

The above described setting gives a nontrivial extension of the underlying frameworks in e.g. \cite{Del}, \cite{Bar-Lau}, \cite{Nor-Sav} and \cite{Nor-qua}. In e.g. \cite{Del} and \cite{Bar-Lau}, the reference filtration $\Fb$ is assumed to be independent of the insurance internal filtration $\Hb$ generated by the non-life portfolio movement. The interaction between the financial and the insurance markets is thus captured only by means of interest rate and/or inflation risk. On the contrary, in e.g. \cite{Nor-Sav} and \cite{Nor-qua}, it is assumed that $\Gb = \Hb = \Fb$. Financial products used for hedging purpose are in these cases liquidly traded catastrophe derivatives and/or reinsurance contracts, which share similar risk structure of the target non-life insurance portfolio.
Considering a more general setting, where $\Fb$ and $\Hb$ are not necessarily independent or equal, is technically challenging, as we discuss in Section \ref{sec: relation with compensator}.
However, the extended reduce-form framework proposed in this paper allows to consider a nontrivial dependence structure between filtrations $\Fb$ and $\Hb$ and still to derive analytical pricing formulas for non-life insurance liabilities. Furthermore, beside the financial instruments used in e.g. \cite{Del}, \cite{Bar-Lau}, \cite{Nor-Sav} and \cite{Nor-qua}, it is possible to use intensity $\mu$ related derivatives as hedging instrument, see discussion in Section \ref{sec: hybrid market}. This last type of derivatives is still not common but is potentially attractive for covering systematic risks arising from non-catastrophe non-life insurance.

\section{Valuation formulas}\label{sec: preliminary results}
In this Section, we state several results within the framework presented in Section \ref{sec: general framework}, i.e. under the structure Assumptions \ref{ass: homo and indep}, \ref{ass: delay} and \ref{ass: accident times}. We follow Section 5.1 of \cite{Bie-Rut} for the presentation. These preliminary calculations are fundamental for providing pricing formulas of non-life insurance claims in Section \ref{sec: pricing non-life}.

We start with an extension of relation (\ref{eq: cond exp tau 0}) and the $\Fb$-independence (\ref{eq: F independent tau 0}) of $\tau^i_0$, $i=1,...,n$.
In particular, if these relations hold for the filtrations $\Hb^{i,0}$, $i=1,...,n$, then they also hold for the filtrations $\Hb^i$, $i = 1,...,n$.

\begin{lemma}\label{lemma: H as F indep}
	For any $t \in [0, \infty]$ and $l,k = 1,...,n$ with $l \neq k$, the $\sigma$-algebras $\H^l_t$ and $\H^k_t$ are $\F_t$-independent. 
	Furthermore, for any $0 \leqslant s \leqslant t \leqslant \infty$ and $i = 1,...,n$, if $X$ is $\H^i_s$-measurable, then
	$\left. E \left[ X \right| \F_t\right] = \left. E \left[ X \right| \F_s \right]$.
\end{lemma}

\begin{proof}
	The proof is straightforward in view of Lemma \ref{lemma: representation Hi}, Remark \ref{rem: H0 are F independent} and the independence conditions in Assumption \ref{ass: homo and indep}. 
	We outline the main steps for the first part of the Lemma for completeness, the second part can be shown in a similar manner. See also proof of Lemma 3.2.1 and Lemma 3.2.2 in \cite{Zhang}.\\
	In view of Lemma \ref{lemma: representation Hi}, (\ref{eq: specif MP 1}) and (\ref{eq: specif MP 2}), for any $t \in [0, \infty]$, $l,k = 1,...,n$ with $l \neq k$, the $\sigma$-algebras $\H^l_t$ and $\H^k_t$ are $\F_t$-independent if and only if 
	\begin{align*}
		&\left. E \left[ \Ind{\tau^l_0 + \theta^l + \tilde{\tau}^l_p \leqslant s} \Ind{\tilde{X}^l_p \in B^l} \Ind{\tau^k_0 + \theta^k + \tilde{\tau}^k_q \leqslant r} \Ind{\tilde{X}^k_p \in B^k} \right| \F_t \right]\\
		=& \left. E \left[ \Ind{\tau^l_0 + \theta^l + \tilde{\tau}^l_p \leqslant s} \Ind{\tilde{X}^l_p \in B^l} \right| \F_t \right] \left. E \left[ \Ind{\tau^k_0 + \theta^k + \tilde{\tau}^k_q \leqslant r} \Ind{\tilde{X}^k_p \in B^k} \right| \F_t \right],
	\end{align*}
	where $s,r \in [0,t] \cap [0, \infty)$ (note that $t$ may assume $\infty$.), $B^l, B^k \in \mathcal{B}(\R_+)$ and $p, q \in \N_0$. Without loss of generality, we take $p \neq 1$ and $q \neq 1$. We consider the following deterministic functions
	\[
		f^l (x) := E \left[  \Ind{\theta^l + \tilde{\tau}^l_p \leqslant s - x} \Ind{\tilde{X}^l_p \in B^l} \right], \ \ \ \ f^l (x) = f^l (x)\Ind{x \leqslant s},
	\]
	\[
		f^k (x) := E \left[  \Ind{\theta^k + \tilde{\tau}^k_q \leqslant r - x} \Ind{\tilde{X}^k_q \in B^k} \right], \ \ \ \ f^k (x) = f^k (x)\Ind{x \leqslant r}.
	\]
	Hence, $f^l (\tau^l_0)$ and $f^k (\tau^k_0)$ are respectively $\H^{l,0}_t$- and $\H^{k,0}_t$-measurable. This together with Remark \ref{rem: H0 are F independent} and the independence conditions in Assumption \ref{ass: homo and indep} leads to the	$\F_t$-independence of $\H^l_t$ and $\H^k_t$,
	\begin{align*}
		&\left. E \left[ \Ind{\tau^l_0 + \theta^l + \tilde{\tau}^l_p \leqslant s} \Ind{\tilde{X}^l_p \in B^l} \Ind{\tau^k_0 + \theta^k + \tilde{\tau}^k_q \leqslant r} \Ind{\tilde{X}^k_p \in B^k} \right| \F_t \right]\\
		=&\left. E \left[ \left. E \left[  \Ind{\tau^l_0 + \theta^l + \tilde{\tau}^l_p \leqslant s} \Ind{\tilde{X}^l_p \in B^l} \Ind{\tau^k_0 + \theta^k + \tilde{\tau}^k_q \leqslant r} \Ind{\tilde{X}^k_p \in B^k} \right| \F_t \vee \sigma(\tau_0^l) \vee \sigma(\tau_0^k) \right] \right| \F_t \right]\\
		=&\left.  E \left[ \left. E \left[  \Ind{x + \theta^l + \tilde{\tau}^l_p \leqslant s} \Ind{\tilde{X}^l_p \in B^l} \Ind{y + \theta^k + \tilde{\tau}^k_q \leqslant r} \Ind{\tilde{X}^k_p \in B^k} \right]\right|_{\substack{x = \tau^l_0\\ y = \tau^k_0}} \right| \F_t \right]\\
		=&\left. E \left[ f^l (\tau^l_0) f^k (\tau^k_0) \right| \F_t \right]\\
		=&\left. E \left[ f^l (\tau^l_0)\right| \F_t \right]\left. E \left[  f^k (\tau^k_0) \right| \F_t \right]\\
		=& \left. E \left[ \Ind{\tau^l_0 + \theta^l + \tilde{\tau}^l_p \leqslant s} \Ind{\tilde{X}^l_p \in B_l} \right| \F_t \right] \left. E \left[ \Ind{\tau^k_0 + \theta^k + \tilde{\tau}^k_q \leqslant r} \Ind{\tilde{X}^k_p \in B_k} \right| \F_t \right].
	\end{align*}
\end{proof}

\noindent As a consequence of Lemma \ref{lemma: H as F indep}, the $\Gb$-conditional expectation can be reduced to $\Fb \vee \Hb^i$-conditional expectation in most cases.

\begin{cor}\label{cor:reduction G to H}
	If $0 \leqslant t \leqslant T < \infty$, and $Y$ is an integrable $(\F_T \vee \H^{i}_T)$-measurable random variable, then
	\[
		\left. E \left[ Y \right| \G_t \right] = \left. E \left[ Y \right| \F_t \vee \H^{i}_t \right].
	\]
\end{cor}

\begin{proof}
	It is sufficient to prove the statement for the indicator functions of the form $Y = \mathbf{1}_A \mathbf{1}_B$ where $A \in \F_T$ and $B \in \H^i_T$. Given $C \in \F_t$, $D^j \in \H^j_t$, $j = 1,...,n$, it follows from Lemma \ref{lemma: H as F indep} that 
	\[
		\int_{C \cap D^1 \cap ... \cap D^n}  \mathbf{1}_A \mathbf{1}_B \ud P = \int_{C \cap D^1 \cap ... \cap D^n}   \left. E \left[ \mathbf{1}_A \mathbf{1}_B \right| \F_t \vee \H^{i}_t \right] \ud P.
	\]
	Details can also be found in the proof of Corollary 3.2.3 in \cite{Zhang}.
\end{proof}

\noindent An other important consequence of Lemma \ref{lemma: H as F indep} is the so called \emph{$H$-hypothesis} between filtrations $\Fb$ and $\Gb$, i.e. the property that every $\Fb$-martingale is also a $\Gb$-martingale.

\begin{cor}
	The $H$-hypothesis holds between filtrations $\Fb$ and $\Gb$.
\end{cor}

\begin{proof}
	By Lemma 6.1.1 of \cite{Bie-Rut}, it is  equivalent to the $H$-hypothesis between two filtrations $\Fb \subseteq \Gb$ that for any $t \geqslant 0$ and any bounded, $\G_t$-measurable random variable $\eta$, it holds
	\begin{equation}\label{eq: eta}
		E\left. \left[\eta \right| \F_\infty \right] = E \left. \left[\eta \right| \F_t \right].
	\end{equation}
	It is sufficient to prove (\ref{eq: eta}) for indicator functions of the form $\mathbf{1}_A \mathbf{1}_{B^1} \ldots \mathbf{1}_{B^n}$, where $A \in \F_t$, $B^i \in \H^i_t$, $i = 1,...,n$. This can be achieved by applying Lemma \ref{lemma: H as F indep} multiple times to get
	\begin{align*}
		E\left. \left[\mathbf{1}_A \mathbf{1}_{B^1} \ldots \mathbf{1}_{B^n} \right| \F_\infty \right] 
		&=  E\left. \left[\mathbf{1}_A \mathbf{1}_{B^1} \ldots \mathbf{1}_{B^n} \right| \F_t \right].
	\end{align*}
\end{proof}

Now we would like to derive some more explicit representations. We note that for every integrable random variable $Y$, $t \geqslant 0$, $i = 1,...,n$ and $j \in \N_0$, we have the decomposition
\begin{equation}\label{eq: decomp Y}
	\left. E \left[  Y  \right| {\H}^{i}_t \vee \F_t \right] = \left. E \left[ \Ind{\tau^i_j > t} Y  \right| {\H}^{i}_t \vee \F_t \right] + \left. E \left[ \Ind{\tau^i_j \leqslant t} Y  \right| {\H}^{i}_t \vee \F_t \right].
\end{equation}
In the following we will evaluate separately the two components on the right-hand side of (\ref{eq: decomp Y}).
The following lemma is important for deriving a representation of the first component.

\begin{lemma}\label{lemma: sigma algebra with tau ij > t}
	For any $t \geqslant 0$, $i = 1,...,n$ and $j \in \N_0$, we have 
	\[
		\H^i_t  \vee \F_t \subseteq \G^{i,j}_t,
	\]
	where
	\begin{equation}\label{eq:G ij}
		\G^{i,j}_t := \left\{ A \in \G : \exists C \in \H^{i, <j}_t \vee \F_t, A \cap \{\tau^i_j > t \} = C \cap \{\tau^i_j > t \} \right\}.
	\end{equation}
\end{lemma}

\begin{proof}
	By Corollary \ref{cor: representation Hi}, it holds that
	\[
		\H^i_t = \H^{i, < j}_t \vee \H^{i, \geqslant j}_t. 
	\]
	Hence, it is sufficient to show that both ${\H}^{i, \geqslant j}_t$ and ${\H}^{i, < j}_t \vee \F_t$ belong to $\G^{i,j}_t$. In the first case, if $i > 1$ and $A = \{\tau^i_k \leqslant s\} \cap \{ X^i_k \in B\}$ for some $k \geqslant j$, $0 \leqslant s \leqslant t$ and $B \in \mathcal{B}(\R)$, we take $C = \emptyset$. Similarly for $i=1$ and $A = \{\tau^1_k \leqslant s\} \cap \{(\theta_k, X^1_k) \in B\}$ for $k \geqslant j$, $0 \leqslant s \leqslant t$ and $B \in \mathcal{B}(\R^2_+)$. In the second case, if $A \in {\H}^{i, < j}_t \vee \F_t$ we take $C = A$.
\end{proof}

\noindent The following Proposition gives two representations of the first component on the right-hand side of (\ref{eq: decomp Y}). Representation (\ref{eq: representation}) is analogue to Lemma 5.1.2. in \cite{Bie-Rut}, representation (\ref{eq: representation 2}) is new and will be used for our further discussion. 

\begin{prop}\label{prop: ind Y}
	For any $t \geqslant 0$, $i = 1,...,n$, $j \in \N_0$ and any integrable $\G$-measurable  random variable $Y$, we have 
	\begin{align}
		\left. E \left[ \Ind{\tau^i_j > t} Y  \right| {\H}^{i}_t \vee \F_t \right] &= \Ind{\tau^i_j > t} \frac{\left. E \left[ \Ind{\tau^i_j > t} Y  \right| \H^{i, <j}_t \vee \F_t \right]}{\left.P \left( \tau^i_j > t \right| \H^{i, <j}_t \vee \F_t \right)} \label{eq: representation} \\
		&= \Ind{\tau^i_j > t} \left. E \left[ Y  \right| \H^{i, \leqslant j}_t \vee \F_t \right]. \label{eq: representation 2}
	\end{align}
\end{prop}

\begin{proof}
	Equality (\ref{eq: representation}) is equivalent to
	\begin{align*}
		\left. E \left[ \Ind{\tau^i_j > t} Y \left.P \left( \tau^i > t \right| \H^{i, <j}_t \vee \F_t\right)  \right| {\H}^{i}_t \vee \F_t \right] =  \Ind{\tau^i_j > t} \left. E \left[ \Ind{\tau^i_j > t} Y  \right| \H^{i, <j}_t \vee \F_t \right].
	\end{align*}
	We note that the right-hand side is $({\H}^{i}_t \vee \F_t)$-measurable. Hence, it suffices to  show that for any $A \in {\H}^{i}_t \vee \F_t$,
	\begin{align*}
		\int_A \Ind{\tau^i_j > t} Y \left.P \left( \tau^i_j > t \right|\H^{i, <j}_t \vee  \F_t \right) \ud P= \int_A \Ind{\tau^i_j > t} \left. E \left[ \Ind{\tau^i_j > t} Y  \right|\H^{i, <j}_t \vee  \F_t \right] \ud P.
	\end{align*}
	By Lemma \ref{lemma: sigma algebra with tau ij > t}, there is an event $C \in \H^{i, <j}_t \vee \F_t$ such that 
	\[
		A \cap \{ \tau^i_j > t \} = C \cap \{ \tau^i_j > t \},
	\]
	 hence
	\begin{align*}
		&\int_A \Ind{\tau^i_j > t} Y \left.P \left( \tau^i_j > t \right| \H^{i, <j}_t \vee \F_t \right) \ud P\\
		=&\int_C \Ind{\tau^i_j > t} Y \left.P \left( \tau^i_j > t \right|\H^{i, <j}_t \vee  \F_t \right)  \ud P \\
		=&\int_C \left.E \left[ \Ind{\tau^i_j > t} Y \right|\H^{i, <j}_t \vee  \F_t \right] \left.E \left[ \Ind{\tau^i_j > t} \right|\H^{i, <j}_t \vee  \F_t \right]  \ud P \\
		=&\int_C \left.E \left[ \Ind{\tau^i_j > t} \left.E \left[ \Ind{\tau^i_j > t} Y \right|\H^{i, <j}_t \vee  \F_t \right] \right| \H^{i, <j}_t \vee \F_t \right]  \ud P \\
		=&\int_C \Ind{\tau^i_j > t} \left.E \left[ \Ind{\tau^i_j > t} Y \right|\H^{i, <j}_t \vee  \F_t \right]  \ud P \\
		=& \int_A \Ind{\tau^i_j > t} \left. E \left[ \Ind{\tau^i_j > t} Y  \right| \H^{i, <j}_t \vee \F_t \right] \ud P.
	\end{align*}
	Equality (\ref{eq: representation 2}) can be proved in the same way. We only need to observe that 
	\[
		\G^{i,j}_t \subseteq \left\{ A \in \G : \exists C \in \H^{i, \leqslant j}_t \vee \F_t, A \cap \{\tau^i_j > t \} = C \cap \{\tau^i_j > t \} \right\}.
	\]
	Hence, the $\sigma$-algebra $\H^{i,<j}_t$ in (\ref{eq: representation}) can be replaced by $\H^{i,\leqslant j}_t$. This concludes the proof.
\end{proof}

\noindent Now we focus on the second component on the right-hand side of (\ref{eq: decomp Y}). The following lemma gives a slightly more general result.

\begin{lemma}\label{lemma: Ind Y tau ij smaller than t}
	For any $t \geqslant 0$, $i = 1,...,n$, $j \in \N_0$, any $\sigma$-algebra $\mathcal{A} \subseteq \G$ and any integrable $\G$-measurable  random variable $Y$, we have
	\[
		 \left. E \left[ \Ind{\tau^i_j \leqslant t} Y  \right| \H^{i,\leqslant j}_t \vee \mathcal{A} \right] =  \left. E \left[ \Ind{\tau^i_j \leqslant t} Y  \right| \H^{i,\leqslant j}_\infty \vee \mathcal{A} \right].
	\]
\end{lemma}

\begin{proof}
	We note that the left-hand side is $(\H^{i,\leqslant j}_\infty \vee \mathcal{A})$-measurable. Since the marked point process $(\tau^i_j, \Theta^i_j)_{j \in \N_0}$ is simple, i.e. the strict monotonicity (\ref{eq: simple}) holds, if $A \in \H^{i,\leqslant j}_\infty \vee \mathcal{A}$, then $A \cap \{\tau^i_j \leqslant t\} \in \H^{i,\leqslant j}_t \vee \mathcal{A}$, and
	\begin{align*}
		\int_{A} \Ind{\tau^i_j \leqslant t} Y  \ud P &= \int_{A \cap \{\tau^i_j \leqslant t\}} Y \ud P = \int_{A \cap \{\tau^i_j \leqslant t\}} \left. E \left[ Y  \right| \H^{i,\leqslant j}_t \vee \mathcal{A} \right] \ud P\\
		&= \int_A \left. E \left[ \Ind{\tau^i_j \leqslant t} Y  \right| \H^{i,\leqslant j}_t \vee \mathcal{A} \right] \ud P.
	\end{align*}
	This concludes the proof.
\end{proof}

\begin{rem}
	 Since we have
\[
	\H^{i,\leqslant j}_\infty = \sigma \left(\tau^i_h, h=1,...,j \right),
\]
Lemma \ref{lemma: Ind Y tau ij smaller than t} shows that, if $\tau^i_j$ has occurred before time $t$, then partial information about $\tau^i_j$ up to $t$ is equivalent to full information about all the random times $\tau^i_h$, $h=1,...,j$. In particular, if $Y$ is a function of $\tau^i_1, ..., \tau^i_j$, i.e. $Y = f (\tau^i_1, ..., \tau^i_j)$, then the conditional expectation is simply
\[
	\left. E \left[ \Ind{\tau^i_j \leqslant t} Y  \right| \H^{i,\leqslant j}_t \vee \mathcal{A} \right] = \Ind{\tau^i_j \leqslant t} Y.
\]
\end{rem}

 We summarize the above results in the following representation theorem.

\begin{theorem}\label{theo: G representation 1}
	For any $t \geqslant 0$, $i = 1,...,n$, $j \in \N_0$ and any integrable $\G$-measurable  random variable $Y$, we have
	\[
		\left. E \left[ Y  \right| {\H}^{i}_t \vee \F_t \right] = \Ind{\tau^i_j \leqslant t} \left. E \left[ Y  \right| \H^{i,\leqslant j}_\infty \vee {\H}^{i, > j}_t \vee \F_t \right] + \Ind{\tau^i_j > t} \left. E \left[ Y  \right| \H^{i, \leqslant j}_t \vee \F_t \right].
	\]
	If furthermore $Y$ is $({\H}^{i}_T \vee \F_T)$-measurable, then 
	\[
		\left. E \left[ Y  \right| \G_t \right] = \Ind{\tau^i_j \leqslant t} \left. E \left[ Y  \right| \H^{i,\leqslant j}_\infty \vee {\H}^{i, > j}_t \vee \F_t \right] + \Ind{\tau^i_j > t} \left. E \left[ Y  \right| \H^{i, \leqslant j}_t \vee \F_t \right].
	\]
\end{theorem}

\begin{proof}
	Since 
	\[
		\left. E \left[ Y  \right| {\H}^{i}_t \vee \F_t \right] = \left. E \left[ \Ind{\tau^i_j \leqslant t} Y  \right| {\H}^{i}_t \vee \F_t \right] + \left. E \left[ \Ind{\tau^i_j > t} Y  \right| {\H}^{i}_t \vee \F_t \right],
	\]
	the first part is a straightforward consequence of Proposition \ref{prop: ind Y} and Lemma \ref{lemma: Ind Y tau ij smaller than t}.
For the second part, it suffices to apply Corollary \ref{cor:reduction G to H}. 
\end{proof}

We now show some results which we will use to solve the reserve estimation problem in Section \ref{sec: pricing non-life}. For this purpose, our approach allows to obtain analytical formulas in a general setting in continuous time, where filtration dependence is taken into account. As we illustrate in Section \ref{sec: relation with compensator}, this is not possible in such a generality by using more classical approaches.
Let $0 \leqslant t \leqslant T < \infty$ and $Z := (Z_t)_{t \in[0,T]}$ be a continuous, bounded and $\Fb$-adapted process. For $i = 1,...,n$, we now consider
\begin{equation}\label{eq: special Y}
	Y = \sum^{\mathbf{N}^i_T}_{j = \mathbf{N}^i_t} X^i_j Z_{\tau^i_j} = \sum^{\infty}_{j = 1} \Ind{t <\tau^i_j \leqslant T} X^i_j Z_{\tau^i_j},
\end{equation}
and compute
\begin{equation}\label{eq: reserve1}
	E\left. \left[ Y \right| \G_t \right] = E\left. \left[ \sum^{\mathbf{N}^i_T}_{j = \mathbf{N}^i_t} X^i_j Z_{\tau^i_j} \right| \G_t \right].
\end{equation}
In particular, similarly to before, we study separately the two components of the decomposition of (\ref{eq: reserve1}) with respect to the first reporting time $\tau^i_1$, i.e.
\begin{equation}\label{eq: 2 components}
	E\left. \left[ \sum^{\mathbf{N}^i_T}_{j = \mathbf{N}^i_t} X^i_j Z_{\tau^i_j} \right| \G_t \right] =  E\left. \left[ \Ind{\tau^i_1 > t} \sum^{\mathbf{N}^i_T}_{j = \mathbf{N}^i_t} X^i_j Z_{\tau^i_j} \right| \G_t \right] + E\left. \left[ \Ind{\tau^i_1 \leqslant t} \sum^{\mathbf{N}^i_T}_{j = \mathbf{N}^i_t} X^i_j Z_{\tau^i_j} \right| \G_t \right],
\end{equation}
 and derive more explicit formulas in terms of the intensity process $\mu$, the distribution of delay $\theta^i$, and the distribution of development $N^i$ after the first reporting. We start with the $\Fb$-conditional expectation of $\tau^i_1$.

\begin{lemma}\label{lemma: cond exp tau1}
	For any $i = 1,...,n$ and $t \geqslant 0$, we have
	\begin{equation}\label{eq: cond exp tau1 with mu with tau1 smaller than t}
		P\left. \left( \tau^i_1 \leqslant t \right| \F_t \right) = \int_0^t G(t - s) e^{-\int^s_0 \mu_v \ud v} \mu_s \ud s,
	\end{equation}
	where $G$ is the cumulative distribution function of $\theta^i$ given in (\ref{eq: delay structure}).
\end{lemma}

\begin{proof}
	Note that by Assumption \ref{ass: homo and indep}, $\theta^i$ is independent of $\F_t \vee \sigma(\tau^i_0)$. Furthermore, both $\theta^i$ and $\tau^i_0$ are $P$-a.s. nonnegative. Therefore, we have
	\begin{align*}
		P\left. \left( \tau^i_1 \leqslant t \right| \F_t \right) &= E\left. \left[ \Ind{\tau^i_0 + \theta^i \leqslant t} \right| \F_t \right] \\
		&= E\left. \left[ E\left. \left[ \Ind{\tau^i_0 \leqslant t}\Ind{\tau^i_0 + \theta^i \leqslant t} \right| \F_t \vee \sigma(\tau^i_0) \right] \right| \F_t \right]\\
		&= E\left. \left[ \Ind{\tau^i_0 \leqslant t} \left. E\left[ \Ind{\theta^i \leqslant t -x} \right]\right|_{x = \tau^i_0} \right| \F_t \right] \\
		&= E\left. \left[  \Ind{\tau^i_0 \leqslant t}  {G}(t - \tau^i_0) \right| \F_t \right].
	\end{align*}
	To conclude we only need to show
	\begin{equation}\label{eq: final step tau1}
		E\left. \left[  \Ind{\tau^i_0 \leqslant t} {G}(t - \tau^i_0) \right| \F_t \right] = \int_0^t {G}(t - s) e^{-\int^s_0 \mu_v \ud v} \mu_s \ud s.
	\end{equation}
	This can be done in the same way as for Proposition 5.1.1 of \cite{Bie-Rut}, in view of relation (\ref{eq: cond exp tau 0}) and the fact that $G$ is continuous according to Assumption \ref{ass: delay}. Firstly, relation (\ref{eq: final step tau1}) can be shown for piecewise constant function $G$ and then it is obtained as a limit for continuous $G$.
\end{proof}

\begin{rem}
	Note that from (\ref{eq: final step tau1}) we can derive the conditional probability that the accident event has been incurred, but not yet reported (IBNR events in the terminology used in the insurance sector).
\end{rem}
	
\noindent In expression (\ref{eq: cond exp tau1 with mu with tau1 smaller than t}) of Lemma \ref{lemma: cond exp tau1}, the parameter $t$ appears also in the integrand. The following corollary improves relation (\ref{eq: cond exp tau1 with mu with tau1 smaller than t}) and shows that the process of conditional expectation $(P\left. \left( \tau^i_1 \leqslant t \right| \F_t \right))_{t \geqslant 0}$ is absolutely continuous with respect to the Lebesgue measure.

\begin{cor}\label{cor: differentiable}
	For any $i=1,...,n$, we have
	\begin{equation}\label{eq: diff P tau 1}
		P\left. \left( \tau^i_1 \leqslant t \right| \F_t \right) = \int_0^t \left( \alpha_0 e^{-\int^s_0 \mu_v \ud v} \mu_s +  \int_0^s g(s - u) e^{-\int^u_0 \mu_v \ud v} \mu_u \ud u \right)\ud s,	
	\end{equation}
	where $\alpha_0$ and $g$ are defined in (\ref{eq: delay structure}).
\end{cor}	

\begin{proof}
	This follows immediately from Assumption \ref{ass: delay} and relation (\ref{eq: cond exp tau1 with mu with tau1 smaller than t}). 	
\end{proof}

\noindent Note that in the following lemma we do not assume that $Z$ is $\Fb$-adapted and the boundedness condition can be generalized.

\begin{lemma}\label{lemma: transf into integral}
	If the process $Z := (Z_u)_{u \in [t, T]}$ is left-continuous and bounded and $Z_t$ is $\F_T$-measurable for all $t \geqslant 0$, then we have
	\begin{align*}
		E\left. \left[ \Ind{t < \tau^i_1 \leqslant T}  Z_{\tau^i_1}  \right| \F_t \right] =  E\left. \left[ \int_t^T  Z_u \ud P\left. \left( \tau^i_1 \leqslant u \right| \F_u \right) \right| \F_t \right],
	\end{align*}
for $i = 1,...,n$ and $t \in [0,T]$. 
\end{lemma}

\begin{proof}
	The argument is similar to Proposition 5.1.1 of \cite{Bie-Rut}, which we outline for completeness. We assume first that both $Z$ and $\tilde{Z}$ are stepwise constant, i.e. without loss of generality, 
	\[
		Z_u = \sum_{j = 0}^n Z_{t_j} \Ind{t_j < u \leqslant t_{j+1}}, \ \ \  \ \tilde{Z}_u = \sum_{j = 0}^n \tilde{Z}_{t_j} \Ind{t_j < u \leqslant t_{j+1}}, 
	\]
	for $t < u \leqslant T$, where $t_0 = t < ... < t_{j + 1} = T$, $Z_{t_j}$ is $\F_{T}$-measurable and $\tilde{Z}_{t_j}$ is independent from $\F_T \vee \sigma(\tau^i_1)$ for all $j = 0,...,n$. By Lemma \ref{lemma: H as F indep}, it holds that 
	 \begin{align}
	 	&E\left. \left[ \Ind{t < \tau^i_1 \leqslant T} \tilde{Z}_{\tau^i_1}  Z_{\tau^i_1}  \right| \F_t \right] \nonumber \\
	 	=& E\left. \left[ \sum_{j = 0}^n E\left. \left[ \tilde{Z}_{t_j} Z_{t_j} \Ind{t_j < \tau^i_1 \leqslant t_{j+1}} \right| \F_{T } \right] \right| \F_t \right] \nonumber\\
	 	=& E\left. \left[ \sum_{j = 0}^n E[\tilde{Z}_{t_j}] Z_{t_j} \left( E\left. \left[ \Ind{\tau^i_1 \leqslant t_{j+1}} \right| \F_{t_{j + 1}} \right] -  E\left. \left[ \Ind{\tau^i_1 \leqslant t_{j}} \right| \F_{t_{j}} \right] \right) \right| \F_t \right] \label{eq: riemann sum gen}\\
	 	=& E\left. \left[ \int_t^T E[\tilde{Z}_u] Z_u \ud P\left. \left( {\tau^i_1 \leqslant u} \right| \F_{u} \right) \right| \F_t \right]. \label{eq: lebesq stielj integral}
	 \end{align}
In the general case, it is sufficient to find stepwise constant approximations for $Z$ and $\tilde{Z}$. Since $Z$ is bounded and $E[\tilde{Z}]$ is continuous and bounded on $[t,T]$, the Riemann sum under the sign of conditional expectation in (\ref{eq: riemann sum gen}) converges to the Lebesgue–Stieltjes integral in expression (\ref{eq: lebesq stielj integral}), hence the convergence of the conditional expectations follows as well by the dominated convergence theorem.
\end{proof}

Now we are able to calculate the first component on the right-hand side of (\ref{eq: 2 components}).
	We define
	\begin{equation}\label{eq: tilde m}
	\begin{array}{*3{>{\displaystyle}l}}
		\tilde m(t) := E  \left[  \sum^{\tilde{\mathbf{N}}_t}_{j = 1} \tilde{X}^i_j  \right], \ \ \ \ &\text{if } t \geqslant 0,\\
		\tilde m(t) := 0, \ \ \ \ &\text{if }t < 0,
	\end{array}	
	\end{equation}
	where $\tilde{\mathbf{N}}$ denotes the ground process of $(\tilde{\tau}^i_j, \tilde{X}^i_j)_{j \in \N_0}$, i.e.
	\begin{equation}\label{eq: ground process tilde}
		\tilde{\mathbf{N}}_t := \sum^\infty_{j=1} \Ind{\tilde{\tau}^i_j \leqslant t}, \ \ \ \ t \geqslant 0.
	\end{equation}
	Note that $\tilde m$ does not depend on $i$ because of Assumption \ref{ass: homo and indep} (2).

\noindent The following the result hold also under different integrability and measurability conditions.

\begin{prop}\label{prop: Y with tau1 bigger than t}
	Let $Z := (Z_t)_{t \in [0,T]}$ be a continuous, bounded and $\Fb$-adapted process and  $Y$ be as in (\ref{eq: special Y}), then for any $t \in [0,T]$,
	\begin{align*}
		&E\left. \left[ \Ind{\tau^i_1 > t}  Y \right| \H^{i}_t \vee \F_t \right]\\
		&=  \Ind{\tau^i_1 > t} \frac{E\left. \left[ \int_t^T \left(E[X^i_1]  Z_u + \int_u^T Z_v  \ud \tilde{m}(v - u)\right) \ud P\left. \left( {\tau^i_1 \leqslant u} \right| \F_{u} \right) \right| \F_t \right]}{P\left. \left( \tau^i_1 > t \right| \F_t \right)}
	\end{align*}
	where $\tilde m$ is defined in (\ref{eq: tilde m}).
\end{prop}

\begin{proof}
	By applying (\ref{eq: representation 2}) in Proposition \ref{prop: ind Y} to $Y$ defined in (\ref{eq: special Y}), we get
	\begin{align}
		&E\left. \left[ \Ind{\tau^i_1 > t}  Y \right| \H^{i}_t \vee \F_t \right] =\Ind{\tau^i_1 > t} E\left. \left[ Y \right| \H^{i,1}_t \vee \F_t \right] \nonumber \\
		=& \Ind{\tau^i_1 > t} E\left. \left[ \sum^{\infty}_{j = 1} \Ind{\tau^i_j \leqslant T}X^i_j Z_{\tau^i_j} \right| \H^{i,1}_t \vee \F_t \right]\nonumber \\
		=& \Ind{\tau^i_1 > t} E\left. \left[ \Ind{\tau^i_1 \leqslant T}X^i_1 Z_{\tau^i_1}  \right| \H^{i,1}_t \vee \F_t \right] + \Ind{\tau^i_1 > t} E \left. \left[ \sum^{\infty}_{j = 2} \Ind{\tau^i_j \leqslant T}X^i_j Z_{\tau^i_j}  \right| \H^{i,1}_t \vee \F_t \right]. \label{eq: right hand side}	
	\end{align}	
	For the first component of (\ref{eq: right hand side}), it is sufficient to use (\ref{eq: representation}) in Proposition \ref{prop: ind Y} and an argument similar to Proposition 5.1.1 of \cite{Bie-Rut}, taking into account the independence condition in Assumption \ref{ass: homo and indep} (3) and Lemma \ref{lemma: H as F indep}. We have hence
	\begin{align*}
	 	 &\Ind{\tau^i_1 > t} E\left. \left[ \Ind{\tau^i_1 \leqslant T}X^i_1 Z_{\tau^i_1}  \right| \H^{i,1}_t \vee \F_t \right]\\
	 	 =& E\left. \left[ \Ind{t < \tau^i_1 \leqslant T}X^i_1 Z_{\tau^i_1}  \right| \H^{i,1}_t \vee \F_t \right]\\
	 	 =&  \Ind{\tau^i_1 > t} \frac{E\left. \left[ \Ind{t < \tau^i_1 \leqslant T}X^i_1 Z_{\tau^i_1}  \right| \F_t \right]}{P\left. \left( \tau^i_1 > t \right| \F_t \right)}\\
	 	=& \Ind{\tau^i_1 > t} \frac{E\left. \left[ \int_t^T E[X^i_1] Z_u \ud P\left. \left( {\tau^i_1 \leqslant u} \right| \F_{u} \right) \right| \F_t \right]}{P\left. \left( \tau^i_1 > t \right| \F_t \right)}.
	 \end{align*}
\noindent Now we focus on the second component of (\ref{eq: right hand side}). We assume first that restricted on the interval $[t,T]$, $Z$ is a bounded, stepwise, $\Fb$-predictable process, i.e. 
	\begin{equation}\label{eq: stepwise Z}
		Z_u = \sum_{k = 0}^n Z_{t_k} \Ind{t_k < u \leqslant t_{k+1}}, 
	\end{equation}
	for $t < u \leqslant T$, where $t_0 = t < ... < t_{n + 1} = T$ and $Z_{t_k}$ is $\F_{t_k}$-measurable for all $k = 0,...,n$.
	In such case, we have
	\begin{align}
		&\Ind{\tau^i_1 > t} E \left. \left[ \sum^{\infty}_{j = 2} \Ind{\tau^i_j \leqslant T}X^i_j Z_{\tau^i_j}  \right| \H^{i,1}_t \vee \F_t \right] \nonumber\\
		=& \Ind{\tau^i_1 > t} E \left. \left[ \sum^{\infty}_{j = 1} \Ind{t < \tau^i_1 + \tilde{\tau}^i_j \leqslant T} \tilde{X}^i_j Z_{\tau^i_j}  \right| \H^{i,1}_t \vee  \F_t \right] \nonumber \\
		=& \Ind{\tau^i_1 > t} E \left. \left[ \sum_{k=0}^n \sum^{\infty}_{j = 1} \Ind{t_k < \tau^i_1 + \tilde{\tau}^i_j \leqslant t_{k + 1}} \tilde{X}^i_j Z_{t_k}  \right| \H^{i,1}_t \vee  \F_t \right] \nonumber\\
		=& \Ind{\tau^i_1 > t} E \left. \left[ \sum_{k=0}^n Z_{t_k}  \left. E \left[  \sum^{\infty}_{j = 1} \Ind{t_k < \tau^i_1 + \tilde{\tau}^i_j \leqslant t_{k + 1}} \tilde{X}^i_j \right| \H^{i,1}_t \vee  \F_{t_k} \vee \sigma(\tau^i_1) \right] \right| \H^{i,1}_t \vee  \F_t \right] \nonumber \\
		=& \Ind{\tau^i_1 > t} E \left. \left[ \sum_{k=0}^n Z_{t_k} \left. E \left[  \sum^{\infty}_{j = 1} \Ind{t_k < x + \tilde{\tau}^i_j \leqslant t_{k + 1}} \tilde{X}^i_j  \right]\right|_{x = \tau^i_1} \right| \H^{i,1}_t \vee  \F_t \right] \nonumber \\
		=& \Ind{\tau^i_1 > t} E \left. \left[ \sum_{k=0}^n Z_{t_k} \left( \tilde{m}(t_{k + 1} - \tau^i_1) - \tilde{m}(t_{k} - \tau^i_1) \right) \right| \H^{i,1}_t \vee  \F_t \right], \label{eq: riemann sum 2}
	\end{align}
	where in the second last equality we use the independence between the marked point process $(\tilde{\tau}^i_j, \tilde{X}^i_j)_{j \in \N_0}$ and the $\sigma$-algebra $\H^{i,1}_\infty \vee \F_\infty$ in Assumption \ref{ass: homo and indep}. This shows that for any bounded, stepwise, $\Fb$-predictable process $Z$, we have
	\[
		\Ind{\tau^i_1 > t} E \left. \left[ \sum^{\infty}_{j = 2} \Ind{\tau^i_j \leqslant T}X^i_j Z_{\tau^i_j}  \right| \H^{i,1}_t \vee \F_t \right] = \Ind{\tau^i_1 > t} E \left. \left[ \int_t^T Z_u \ud \tilde{m}(u - \tau^i_1) \right| \H^{i,1}_t \vee  \F_t \right].
	\]
	A continuous bounded process $Z$ can be approximated by a sequence of bounded, stepwise and $\Fb$-predictable processes, i.e. there is a sequence $Z^n$ of the form (\ref{eq: stepwise Z}) such that
	\[
		Z^n \longrightarrow Z \ \ \ \ \text{and}  \ \ \ \ |Z^n| \leqslant M,
	\]	
	with $M > 0$. 
	Since $\tilde m$ is right-continuous and monotone, the Lebesgue-Stieltjes integral
	\begin{equation}\label{eq: l-s integral}
		\int_t^T Z_u \ud \tilde{m}(u - \tau^i_1)
	\end{equation}
	is well defined. It holds by Lebesgue Theorem  
	\[	
		\int_t^T Z^n_u \ud \tilde{m}(u - \tau^i_1) \longrightarrow \int_t^T Z_u \ud \tilde{m}(u - \tau^i_1).
	\]
	Furthermore,
	\begin{equation}\label{eq: lebesgue conv}
		\left|\int_t^T Z^n_u \ud \tilde{m}(u - \tau^i_1) \right| \leqslant M \left| \int_t^T \ud \tilde{m}(u - \tau^i_1)\right| = M |\tilde{m}(T - \tau^i_1) - \tilde{m}(t - \tau^i_1)|.
	\end{equation}		
	The right-hand side of (\ref{eq: lebesgue conv}) is uniformly bounded by (\ref{eq: tilde m}) and (\ref{eq: finite moment}). By applying again Lebesgue Theorem, we have also the convergence of the conditional expectations
	\[
		\Ind{\tau^i_1 > t} E \left. \left[ \int_t^T Z^n_u \ud \tilde{m}(u - \tau^i_1) \right| \H^{i,1}_t \vee  \F_t \right] \longrightarrow E \left. \left[ \int_t^T Z_u \ud \tilde{m}(u - \tau^i_1) \right| \H^{i,1}_t \vee  \F_t \right].
	\]
	We note that $\tilde m(u) = 0 $ for $u < 0$, hence,
	\begin{align*}
		&\Ind{\tau^i_1 > t} E \left. \left[ \sum^{\infty}_{j = 2} \Ind{\tau^i_j \leqslant T}X^i_j Z_{\tau^i_j}  \right| \H^{i,1}_t \vee \F_t \right]\\
		&= \Ind{\tau^i_1 > t} E \left. \left[ \int_t^T Z_u \ud \tilde{m}(u - \tau^i_1) \right| \H^{i,1}_t \vee  \F_t \right]\\
		&= E \left. \left[ \Ind{t < \tau^i_1 \leqslant T} \int_t^T Z_u \ud \tilde{m}(u - \tau^i_1) \right| \H^{i,1}_t \vee  \F_t \right].
	\end{align*}
	By applying again (\ref{eq: representation}) in Proposition \ref{prop: ind Y} to the above expression, we get
 	\begin{align*}
		&\Ind{\tau^i_1 > t} E \left. \left[ \sum^{\infty}_{j = 2} \Ind{\tau^i_j \leqslant T}X^i_j Z_{\tau^i_j}  \right| \H^{i,1}_t \vee \F_t \right]\\
		&= \Ind{\tau^i_1 > t}\frac{  E \left. \left[ \Ind{t < \tau^i_1 \leqslant T}  \int_t^T Z_u \ud \tilde{m}(u - \tau^i_1) \right| \F_t \right]}{P\left. \left( \tau^i_1 > t \right| \F_t \right)}.
 	\end{align*}
 	Let $\tilde Z_s := \int_t^T Z_u \ud \tilde{m}(u - s)$, $s \in [0,T]$. We note that $\tilde m$ is right-continuous and monotone. On one hand, for fixed $s \in [0,T]$, the function $d_s(u) := \tilde m (u - s)$, $u \in [0,T]$, is also right-continuous and monotone and defines the cumulative distribution function of a finite positive measure, in view of (\ref{eq: finite moment}). On the other hand, for fixed $u \in [0,T]$, the function $\tilde m (u - s)$, $s \in [0,T]$, is left-continuous in $s$, i.e. for every series $s_n \nearrow s$, we have the pointwise convergence 
 	\[
		\lim_{s_n \nearrow s} d_{s_n}(u) = d_s(u) \ \ \ \ \text{for all } u \in [0,T]
 	\]
 	of the cumulative distribution functions, equivalent to the convergence in distribution or weak convergence in measure. Note that a series of positive finite measures $(\nu_n)_{n \in \N}$ converges weakly to a positive finite measure $\nu$, if for all bounded continuous functions $f$ the following holds \[ \int f \ud \nu_n \longrightarrow \int f \ud \nu.  \]This yields the convergence
 	\[
 		\tilde Z_{s_n} \longrightarrow \tilde Z_s, \ \ \ \ P-\text{a.s.},
 	\] 
 	that is, $\tilde Z_s := \int_t^T Z_u \ud \tilde{m}(u - s)$, $s \in [0,T]$, is left-continuous. Furthermore, it is also bounded. Now we apply Lemma \ref{lemma: transf into integral} and obtain
 	\begin{align*}
		&\Ind{\tau^i_1 > t}\frac{  E \left. \left[ \Ind{t < \tau^i_1 \leqslant T}  \int_t^T Z_u \ud \tilde{m}(u - \tau^i_1) \right| \F_t \right]}{P\left. \left( \tau^i_1 > t \right| \F_t \right)}\\
		&=\Ind{\tau^i_1 > t}\frac{  E \left. \left[ \Ind{t < \tau^i_1 \leqslant T}  \tilde Z_{\tau^i_1} \right| \F_t \right]}{P\left. \left( \tau^i_1 > t \right| \F_t \right)}\\
		 &= \Ind{\tau^i_1 > t} \frac{E\left. \left[ \int_t^T \tilde Z_u \ud P\left. \left( {\tau^i_1 \leqslant u} \right| \F_{u} \right) \right| \F_t \right]}{P\left. \left( \tau^i_1 > t \right| \F_t \right)}\\
		 &= \Ind{\tau^i_1 > t} \frac{E\left. \left[ \int_t^T \left(\int_t^T Z_v  \ud \tilde{m}(v - u) \right) \ud P\left. \left( {\tau^i_1 \leqslant u} \right| \F_{u} \right) \right| \F_t \right]}{P\left. \left( \tau^i_1 > t \right| \F_t \right)}.
 	\end{align*}	
 	As the last step, we note that for $u < s$, $\int_t^T Z_u \ud \tilde{m}(u - s) = \int_s^T Z_u \ud \tilde{m}(u - s) $ since $\tilde{m}(u - s) = 0$. This concludes the proof.
\end{proof}

\begin{rem}
	The proof of Proposition \ref{prop: Y with tau1 bigger than t} relies on Assumption \ref{ass: homo and indep}.
	Another sufficient condition would be the continuity of $\tilde m$, such as in the case of a compound Poisson process or a Cox process with continuous intensity process and integrable marks. Indeed, since $\tilde m (u) = 0$ for $u < 0$,
	\begin{align*}
		\Ind{\tau^i_1 > t} \left|\int_t^T Z^n_u \ud \tilde{m}(u - \tau^i_1) \right| &\leqslant \Ind{\tau^i_1 > t} M \left| \int_t^T \ud \tilde{m}(u - \tau^i_1)\right|\\
		&= \Ind{t < \tau^i_1 \leqslant T} M |\tilde{m}(T - \tau^i_1) - \tilde{m}(t - \tau^i_1)|,	
	\end{align*}
	and the right-hand side is uniformly bounded if $\tilde m$ is continuous.
\end{rem}

The following proposition gives a representation of the second component on the right-hand side of (\ref{eq: 2 components}).

\begin{prop}\label{prop: Y with tau 1 smaller than t} 
	Under the same assumptions of Proposition \ref{prop: Y with tau1 bigger than t}, if for each $i= 1,...,n$, the process $\left(\sum^{\tilde{\mathbf{N}}_t}_{j = 1} \tilde{X}^i_j\right)_{t \in[0,T]} $, where $\tilde{\mathbf{N}}$ is defined in (\ref{eq: ground process tilde}),
is of  independent increments with respect to its natural filtration $\tilde{\Hb}^i$, then for $t \in [0,T]$ and $Y$ as in (\ref{eq: special Y}),
	it holds
	\[
		\left. E \left[ \Ind{\tau^i_1 \leqslant t} Y  \right| \H^{i}_t \vee \F_t \right] =\Ind{\tau^i_1 \leqslant t} \left. \left. E \left[ \int_t^T  Z_{u} \ud \tilde{m}({u - x})  \right| \H^{i,1}_\infty \vee \tilde{\H}^{i}_{t - x} \vee \F_t \right] \right|_{x = \tau^i_1},
	\]
	for $i=1,...,n$.
\end{prop}

\begin{proof}
	 It follows from Lemma \ref{lemma: Ind Y tau ij smaller than t} that 
	 \begin{align*}
		 \left. E \left[ \Ind{\tau^i \leqslant t} Y  \right| \H^{i}_t \vee \F_t \right]=  \left. E \left[ \Ind{\tau^i_1 \leqslant t} Y  \right| \H^{i,1}_\infty \vee \H^{i, >1}_t \vee \F_t \right].
	\end{align*}
	As in the proof of Proposition \ref{prop: Y with tau1 bigger than t}, we assume first $Z$ of the form (\ref{eq: stepwise Z}). Similar calculations lead to
	 \begin{align*}
		 &\left. E \left[ \Ind{\tau^i_1 \leqslant t} Y  \right| \H^{i,1}_\infty \vee \H^{i, >1}_t \vee \F_t \right] \nonumber \\
		 =& \Ind{\tau^i_1 \leqslant t} \left. E \left[ \sum^{\infty}_{j = 1} \Ind{t <\tau^i_1 + \tilde{\tau}^i_{j} \leqslant T}\tilde{X}^i_j Z_{\tau^i_j} \right| \H^{i,1}_\infty \vee \H^{i, >1}_t \vee \F_t \right] \nonumber \\
		=& \Ind{\tau^i_1 \leqslant t} \left. \left. E \left[ \sum_{i = 0}^n \sum^{\infty}_{j = 1} \Ind{t_i < x + \tilde{\tau}^i_{j} \leqslant t_{i+1}}\tilde{X}^i_j Z_{t_i}  \right| \H^{i,1}_\infty \vee \H^{i, >1}_t \vee \F_t \right] \right|_{x = \tau^i_1} \nonumber \\  
		=& \Ind{\tau^i_1 \leqslant t} \left. \left. E \left[ \sum_{i = 0}^n \sum^{\infty}_{j = 1} \Ind{t_i < x + \tilde{\tau}^i_{j} \leqslant t_{i+1}}\tilde{X}^i_j Z_{t_i} \right| \H^{i,1}_\infty \vee \tilde{\H}^{i}_{t - x} \vee \F_t \right] \right|_{x = \tau^i_1},
	\end{align*}
	where the last step follows from the definitions of the filtrations.
	Using the tower property, the independence between the marked point process  $(\tilde{\tau}^i_j, \tilde{X}^i_{j})_{j \in \N_0}$ and $\F_\infty \vee \H^{i,1}_\infty$ (see Assumption \ref{ass: homo and indep}), and the independence of increments of the process the process $\left(\sum^{\tilde{\mathbf{N}}_t}_{j = 1} \tilde{X}^i_j\right)_{t \in[0,T]} $, we get furthermore
	\begin{align}
		&\left. E \left[ \Ind{\tau^i_1 \leqslant t} Y  \right| \H^{i,1}_\infty \vee \H^{i, >1}_t \vee \F_t \right] \nonumber \\
		=& \Ind{\tau^i_1 \leqslant t} \left. \left. E \left[ \sum_{i = 0}^n Z_{t_i} \left( \left. E \left[  \sum^{\infty}_{j = 1} \Ind{ \tilde{\tau}^i_{j} \leqslant t_{i+1} - x}\tilde{X}^i_j  \right| \H^{i,1}_\infty \vee \tilde{\H}^{i}_{t - x} \vee \F_{t_{i}} \right] \right. \right. \right. \right. \nonumber \\
		&- \left.\left. \left. \left. \left. E \left[  \sum^{\infty}_{j = 1} \Ind{ \tilde{\tau}^i_{j} \leqslant t_i -x }\tilde{X}^i_j  \right| \H^{i,1}_\infty \vee \tilde{\H}^{i}_{t - x} \vee \F_{t_i} \right] \right) \right| \H^{i,1}_\infty \vee \tilde{\H}^{i}_{t - x} \vee \F_t \right] \right|_{x = \tau^i_1} \nonumber \\		
		=& \Ind{\tau^i_1 \leqslant t} \left. \left. E \left[ \sum_{i = 0}^n  Z_{t_i} \left( \tilde{m}({t_{i+1} - x}) -  \sum^{\tilde{\mathbf{N}}_{t-x}}_{j = 1}\tilde{X}^i_j - \tilde{m}({t_{i} - x}) + \sum^{\tilde{\mathbf{N}}_{t-x}}_{j = 1}\tilde{X}^i_j  \right) \right| \H^{i,1}_\infty \vee \tilde{\H}^{i}_{t - x} \vee \F_t \right] \right|_{x = \tau^i_1}\nonumber \\
		=& \Ind{\tau^i_1 \leqslant t} \left. \left. E \left[ \sum_{i = 0}^n  Z_{t_i} \left( \tilde{m}({t_{i+1} - x}) - \tilde{m}({t_{i} - x})  \right) \right| \H^{i,1}_\infty \vee \F_t \right] \right|_{x = \tau^i_1}. \label{eq: riemann sum}
	\end{align}
	This yields that for any bounded, stepwise, $\Fb$-predictable process $Z$, we have
	\begin{align*}
		\left. E \left[ \Ind{\tau^i_1 \leqslant t} Y  \right| \H^{i}_t \vee \F_t \right] =\Ind{\tau^i_1 \leqslant t} \left. \left. E \left[ \int_t^T  Z_{u} \ud \tilde{m}({u - x})  \right| \H^{i,1}_\infty \vee \F_t \right] \right|_{x = \tau^i_1}.
	\end{align*}
	
	\noindent If $Z$ is continuous, bounded and $\Fb$-adapted, then $Z$ can be approximated by a sequence of bounded, stepwise and $\Fb$-predictable processes. This together with the fact that $\tilde m$ is right-continuous and monotone guarantees that the Riemann sum in (\ref{eq: riemann sum}) under the sign of conditional expectation converges to Lebesgue–-Stieltjes integral, using the same arguments of Proposition \ref{prop: Y with tau1 bigger than t}.
\end{proof}

We summarize the results in the following theorem, which gives an explicit representation of $\Gb$-conditional expectation with respect to the first reporting time $\tau^i_1$. Note that as most of the results, the conclusion also holds under alternative integrability and measurability conditions.

\begin{theorem}\label{theo: G representation 2}
	Let $Z := (Z_t)_{t \in [0, T]}$ be a continuous, bounded and $\Fb$-adapted process, $Y$ be of the form (\ref{eq: special Y}).
	If the process $\left(\sum^{\tilde{\mathbf{N}}_t}_{j = 1} \tilde{X}^i_j\right)_{t \in[0,T]} $, has independent increments and $\tilde m$ is defined in (\ref{eq: tilde m}), then
	\begin{align*}
		E\left. \left[ Y \right| \G_t \right] =& \Ind{\tau^i_1 \leqslant t} \left. \left. E \left[ \int_t^T  Z_{u} \ud \tilde{m}({u - x})  \right| \H^{i,1}_\infty \vee \tilde{\H}^{i}_{t - x} \vee \F_t \right] \right|_{x = \tau^i_1}\\
		& +  \Ind{\tau^i_1 > t} \frac{E\left. \left[ \int_t^T \left(E[X^i_1]  Z_u + \int_u^T Z_v  \ud \tilde{m}(v - u)\right) \ud P\left. \left( {\tau^i_1 \leqslant u} \right| \F_{u} \right) \right| \F_t \right]}{P\left. \left( {\tau^i_1 > t} \right| \F_{t} \right)},
	\end{align*}
	for $i=1,...,n$,
	where
	\[	
		P\left. \left( {\tau^i_1 \leqslant t} \right| \F_{t} \right) = \int_0^t \left( \alpha_0 e^{-\int^u_0 \mu_v \ud v} \mu_u + \int_0^u g(u - v) e^{-\int^v_0 \mu_s \ud s} \mu_v \ud v \right) \ud u,
	\]
	with $\alpha_0$ and $g$ defined in (\ref{eq: delay structure}).
\end{theorem}

\begin{proof}
	It is enough to combine Corollary \ref{cor:reduction G to H}, Lemma \ref{lemma: cond exp tau1}, Corollary \ref{cor: differentiable}, Proposition \ref{prop: Y with tau1 bigger than t} and Proposition \ref{prop: Y with tau 1 smaller than t}.
\end{proof}

Compared to Theorem \ref{theo: G representation 1}, Theorem \ref{theo: G representation 2} is more explicit and has the advantage that the representation is expressed as function of $\mu$, the distribution of $\theta^i$ and the distribution of $(\tilde{\tau}^i_j, \tilde{X}^i_j)_{j \in \N_0}$. This result will be useful for the concrete reserving problem in hybrid market in Section \ref{sec: hybrid market}.

\section{Comparison with the compensator approach}\label{sec: relation with compensator}

In this section, we compare our framework with the compensator approach for non-life insurance  in the existing literature. Within this section, the filtration $\Hb$ denotes the natural filtration of a marked point process $(\tau_n, X_n)_{n \in \N_0}$, with marked cumulative process $N$, and $\Gb$ is a generic enlargement of $\Hb$. We set $\H := \H_\infty$ and $\G := \G_\infty$.

In most of the current literature, e.g. \cite{Bar-Lau}, \cite{Nor-Sav}, \cite{Nor-qua} and \cite{Schm}, the study of non-life insurance contracts is based on modeling the $\Gb$-compensator of $N$, since the $\Gb$-compensator is involved in the pricing formula and in the calculation of the hedging strategy. In the reduced-form framework for life insurance, the direct modeling approach and the compensator approach coincide, see e.g. \cite{Bie-Rut}. However, the compensator approach presents several difficulties in a non-life insurance setting with nontrivial filtrations' dependence.

\begin{defn}
	The \emph{$\Gb$-mark-predictable $\sigma$-algebra} on the product space $\R_+ \times \mathcal{B}(\R_+) \times \Omega$ is the $\sigma$-algebra generated by sets of the form $(s,t] \times B \times A$ where $0 < s < t$, $B \in \mathcal{B}(\R_+)$ and $A \in \G_s$.
\end{defn}
\begin{defn}
	The $\Gb$-compensator of a marked point process $(\tau_n, X_n)_{n \in \N_0}$ is any $\Gb$-mark-predictable, cumulative process $\Lambda(t, B, \omega)$ such that, $(\Lambda(t,B))_{t \geqslant 0}$ with $\Lambda(t,B)(\cdot) := \Lambda(t, B, \cdot)$ is the $\Gb$-compensator of the point process $(N(t, B))_{t \geqslant 0}$. We use the notation $({\Lambda}_t)_{t \geqslant 0}$, ${\Lambda}_t := \Lambda(t, \R_+)$, to denote the $\Gb$-compensator of the ground process $(\mathbf{N}_t)_{t \geqslant 0}$.
\end{defn}

\noindent Theorem 14.2.IV(a) of \cite{Dal} shows that given a marked point process $(\tau_n, X_n)_{n \in \N_0}$ with finite first moment measure, its $\Gb$-compensator $\Lambda$ always exists and is $(l \otimes P)$-a.e. unique, where $l$ denotes the Lebesgue measure on $\R_+$. In particular, for all $(t, B, \omega) \in \R_+ \times \mathcal{B}(\R_+) \times \Omega$, the following relation holds
\begin{equation}\label{eq: relation lambda}
	\Lambda(t, B, \omega) = \int_0^t \kappa(B| s, \omega) {\Lambda}(\ud s, \omega),
\end{equation}
where $\kappa(B| s, \omega)$, $B \in \mathcal{B}(\R_+)$, $s \geqslant 0$, $\omega \in \Omega$, is the unique predictable kernel such that for all $A \in \G_s, 0 < s < t, B \in \mathcal{B}(\R_+)$,
\[
	\int_A \int_s^t N(u, B) (\omega) \ud u P(\ud \omega) = \int_A \int_s^t \kappa(B| u, \omega) {\mathbf{N}}_u( \omega) \ud u P(\ud \omega).
\] 

However, under general conditions it is not always true that given a $\Gb$-mark-predictable and cumulative process $\Lambda$, there exists a marked point process $(\tau_n, X_n)_{n \in \N_0}$ with $\Gb$-compensator $\Lambda$. The problem is first mentioned in \cite{Jacod}, where the case with $\Gb = \Hb$ is solved. An extension of the existence theorem to the case of $\Gb = \Fb \otimes \Hb$, i.e. when the filtrations $\Fb$ and $\Hb$ are independent, is provided in \cite{Det}.
Furthermore while the law of $N$ is uniquely determined by the $\Hb$-compensator, this is not true for the $\Gb$-compensator.
 See discussion in  \cite{Jacod} and Section 4.8 of \cite{Jac}. 
Consequently, the literature with the compensator approach is mostly limited to the cases of $\Gb \equiv \Hb$, see e.g. \cite{Nor-Sav}, \cite{Nor-qua}, or $\Gb = \Fb \otimes \Hb$, see e.g. \cite{Bar-Lau}.

In the following we provide a sufficient condition in the general case of $\Gb = \Fb \vee \Hb$, such that the law of $N$ is uniquely determined by $\Lambda$. Similarly to e.g. \cite{Nor-Sav} and \cite{Nor-qua}, we assume that the $\Gb$-compensator of $(\tau_n, X_n)_{n \in \N_0}$ has the following form
\begin{equation}\label{eq: compensator with intensity}
	\Lambda(t,B) = \int_0^t \int_B \lambda_s \eta_s (\ud x) \ud s \ \ \ \ \text{for all } t \geqslant 0, \ B \in \mathcal{B}(\R_+),
\end{equation}
where $\lambda := (\lambda_t)_{t \geqslant 0}$ is a $\Gb$-progressively measurable process and the mapping $\eta$
\begin{align*}
	\eta:  \R_+ \times \mathcal{B}(\R_+) \times \Omega  &\longrightarrow (\R_+, \mathcal{B}(\R_+))\\
	(t, B, \omega) &\mapsto \eta_t(B)(\omega),
\end{align*}
is such that for every $t \geqslant 0$, $\omega \in \Omega$, $\eta(t,\cdot,\omega)$ is a probability measure on $(\R_+, \mathcal{B}(\R_+))$, and for every $B \in \mathcal{B}(\R_+)$, $(\eta_t(B))_{t \geqslant 0}$ is a $\Gb$-progressively measurable process. 
Clearly, we have
\[
	{\Lambda}_t = \int_0^t \lambda_s \ud s \ \ \ \ \text{for all } t \geqslant 0.
\]
In particular, we can choose a predictable version of both $\lambda$ and $\eta$, see Section 14.3 of \cite{Dal} for details. The processes $\lambda$ and $\eta$ can be interpreted respectively as jump intensity and jump size intensity. We recall that a marked point process $(\tau_n, X_n)_{n \in \N_0}$ has \emph{independent marks} if the marks $(X_n)_{n \in \N_0}$ are mutually independent given $\mathbf{N}$. 

\begin{prop}\label{prop: intensity uniquely determines law}
	The law of a simple marked point process $(\tau_n, X_n)_{n \in \N_0}$ on $(\Omega, \H)$ with finite first moment measure, independent marks and of the form (\ref{eq: compensator with intensity}) is uniquely determined by $\lambda$ and $\eta$. If furthermore $\lambda$ is $\Hb$-measurable, then also the law of $N$ on $(\Omega, \G)$ is uniquely defined.
\end{prop}

\begin{proof}
	By Proposition 6.4.IV(a) of \cite{Dal}, the law of marked point process with independent marks is uniquely determined by the kernel $\kappa$ and the distribution of ${N}$. 
	According to relations (\ref{eq: relation lambda}) and (\ref{eq: compensator with intensity}), the kernel $\kappa$ is given by 
	\[
		\kappa(B|t, \omega) = \eta_t(B)(\omega), \ \ \ \ (t, B, \omega) \in \R_+ \times \mathcal{B}(\R_+) \times \Omega.
	\]
	Corollary 4.8.5 of \cite{Jac} and Theorem 14.2.IV(c) of \cite{Dal} show that, if $N$ is simple and of the form (\ref{eq: compensator with intensity}), the process $(E[\lambda_t|\H_t])_{t \geqslant 0}$ determines uniquely the distribution of ${N}$ on $(\Omega, \H)$. If in addition $\lambda$ is $\Hb$-adapted, then by Theorem 4.8.1 of \cite{Jac}, also the distribution of ${N}$ on $(\Omega, \G)$ is uniquely determined.
\end{proof}

\noindent Nevertheless, Proposition \ref{prop: intensity uniquely determines law} requires the jump intensity process $\lambda$ to be $\Hb$-adapted in order to have $N$ uniquely defined in law, which is an unnatural condition in our context.

On the contrary, the approach proposed in Section \ref{sec: general framework} allows to take into account a dependence structure between the filtrations $\Hb$ and $\Gb$ by directly modeling the $\Fb$-adapted intensity process $\mu$. Furthermore, this allows to obtain analytical results for valuation formulas as shown in Section \ref{sec: preliminary results}.

\section{Pricing and hedging of non-life insurance liability cash flow in hybrid market}\label{sec: hybrid market}

In this section, we address the issue of pricing and hedging non-life insurance liability cash flows by applying the results of Section \ref{sec: preliminary results} under the benchmarked risk-minimization approach. We assume that $P$ is the real world measure and consider a general structure for a hybrid insurance and financial market. We fix a time horizon $T$ with $0 < T < \infty$, and denote the inflation index process by $I:=(I_t)_{t \in [0,T]}$, which represents the percentage increments of the Consumer Price Index (CPI) and follows a nonnegative $(P, \Fb)$-semimartingale. 
We distinguish real price value, i.e. inflation adjusted, from nominal price value, which can be converted in real value at any time $t \in [0,T]$, if divided by the inflation index $I_t$. If not otherwise specified, all price values are expressed in nominal value. 

We consider $d$ liquidly traded primary assets on the financial market described by price process vector $S := (S^1_t, ..., S^d_t)_{t \in [0, T]}$, which follows a real-valued $(P, \Fb)$-semimartingale. 
We assume that there is a publicly accessible index, based on the intensity process $\mu$ and modelled by the process $L:=(L_t)_{t \in [0,T]}$ with
\[
	L_t := e^{-\Gamma_t}, \ \ \ \ t \in [0,T],
\]
see e.g. \cite{Cai-Bla}. This index reflects the underlying systematic risk-factor related to the insurance portfolio, such as mortality risk, weather risk, car accident risk, etc.
We distinguish three kinds of primary assets as elements of the vector $S$:
\begin{enumerate}	
	\item financial assets, such as the zero-coupon bond, call and put options, futures etc.;
	\item inflation linked derivatives, such as inflation linked zero-coupon bond (called also zero-coupon Treasury Inflation Protected Security, TIPS), which pays off $I_T$ (equivalent to 1 real unit) at time $T$, inflation linked call and put options, etc.;
	\item macro risk factor linked derivatives based on the index $L$, such as longevity bond which pays off $L_T$ at time $T$, weather index-based derivatives, etc. Note that apart from longevity bonds, other macro index based derivatives are still not common in the market.
\end{enumerate}
We denote by $L(S,P,\Gb)$ the space of $\R^d$-valued $\Gb$-predictable $S$-integrable processes. 
By following the definitions in \cite{Bia-Pra}, we call portfolio or value process $S^\delta:=(S^\delta_t)_{t \in [0,T]}$ associated to a trading strategy $\delta:= (\delta_t)_{t \in [0,T]}$ in $L(S,P,\Gb)$ the following càdlàg optional process
	\[
		S^\delta_{t-} = \delta_t^\top S_t = \sum_{i=1}^d \delta^i_t S^i_t, \ \ \ \ t \in [0,T].
	\]
It is called self-financing if 
	\[
		S^\delta_t = S_0^\delta + \int_0^t \delta_{u-}^\top \ud S_u = S_0^\delta + \sum_{i=1}^d \int_0^t \delta^i_{u-} \ud S^i_u, \ \ \ \ t \in [0,T].
	\]
We introduce the following set
\[
	\mathcal{V}^{+}_x = \{ S^\delta \text{ self-financing }: \ \delta \in L(S, \P, \Gb), \ S_0^\delta = x > 0, \ S^\delta > 0\}.
\]

\begin{defn}\label{def: benchmark}
	A \emph{benchmark} or \emph{numéraire portfolio} $S^{*}:=(S^{*}_t)_{t \in [0,T]}$ is an element of $\mathcal{V}^{+}_1$, such that 
	\[
		\frac{S^\delta_s}{S^{*}_s} \geqslant E \left[ \left. \frac{S^\delta_t}{S^{*}_t} \right| \G_s \right], \ \ \ \ s,t \in [0,T], \ \ t \geqslant s.
	\]
\end{defn}

We follow the approach of \cite{Pla3} and work under the following assumption.

\begin{assump}\label{ass: benchmark}
	There exists a benchmark portfolio $S^*$.
\end{assump}

In \noindent \cite{Hul-Sch}, it is shown that Assumption \ref{ass: benchmark} is weaker than assuming the existence of an equivalent martingale measure. As discussed in \cite{Bia}, this weak no-arbitrage assumption is more suitable for modeling a hybrid market as in our case and allows to work directly under the real world measure $P$. An extensive background of the benchmark approach and its application can be found in \cite{Pla3}. Note that, as discussed in \cite{Pla3}, the benchmark portfolio $S^*$ can be identified with a sufficiently diversified portfolio such as Morgan Stanley capital weighted world stock accumulation index, i.e. MSCI world index.

\noindent Given a generic random variable or process $X$, we denote by $\hat X:= X/S^*$ the \emph{benchmarked} value of $X$.
The following lemma is proved in \cite{Bia-Cre}.

\begin{lemma}\label{lemma: benchmarked primary assets}
	If the vector process of primary assets $S$ is continuous, then the benchmarked vector process $\hat S := {S}/{S^*}$ is a $(P, \Gb)$-local martingale.
\end{lemma}

\noindent For the sake of simplicity, we assume the following conditions similar to the ones in \cite{Bia-Zha}.
\begin{assump}\label{ass: concrete model}
	The inflation index process $I = (I_t)_{t \in [0,T]}$ and the vector process of primary assets $S$ are continuous. The benchmark portfolio $S^* = (S^*_t)_{t \in [0,T]}$ is continuous, $\Fb$-adapted, and the benchmarked value process $\hat{S} := S/S^*$ is an $(\Fb, P)$-local martingale.  Inflation linked zero-coupon bond (or TIPS) is a primary asset, i.e. an element of the vector $S$.
\end{assump}

The payment stream in real unit of the insurance company towards policyholders is modelled by a nonnegative $(P, \Gb)$-semimartingale $D := ( D_t)_{t \in [0, T]}$. We denote by $A:=(A_t)_{t \in [0,T]}$ the nominal benchmarked cumulative payment, namely
\begin{equation}\label{eq: A process definition}
	A_t := \int^t_0 \frac{I_u}{S^*_u} \ \ud D_u, \ \ \ \ t \in [0,T].
\end{equation}

\begin{defn}
	We call \emph{real world pricing formula} associated to $A$ the following formula
	\begin{equation}\label{eq: pricing formula dividend}
		V_t := \frac{S^*_t}{I_t} \condespg{A_T - A_t} = \frac{S^*_t}{I_t} \condespg{\int_{]t, T]} \frac{I_u}{S^*_u} \ \ud D_u},
	\end{equation}
	for $t \in [0, T]$.
\end{defn}

\noindent Here $V_t$ in (\ref{eq: pricing formula dividend}) is expressed in real value, i.e. inflation adjusted value. According to the benchmark approach of \cite{Pla3}, a portfolio's process is fair, if its benchmarked value process is a $P$-martingale. The real world pricing formula (\ref{eq: pricing formula dividend}) then provides the fair portfolio of minimal price among all replicating self-financing portfolios for a given benchmarked claim $\hat H$, if $\hat H$ is hedgeable. 
In the case of incomplete market models, it corresponds to the \emph{benchmarked risk-minimizing price} for the payment process $A$ at time $t$, if $A$ is square integrable, i.e.
\[
	\sup_{t \in [0,T]} \esp{A_t^2} < \infty.
\]
For hedging purpose, we now use the relation between benchmark approach and risk minimization, illustrated in \cite{Pla3} and \cite{Bia-Cre} for a single payoff and in Appendix of \cite{Bia-Zha} for the case of dividend payments. In particular, Theorem A.7 of \cite{Bia-Zha} shows that if 
\begin{equation}\label{eq: GKW decomposition}
	A_T = \mathbb{E}[A_T] + \int_0^T \left( \delta^A_u \right)^\top \ud \hat S_u + L^A_T, \ \ \ \ P-\text{a.s.},
\end{equation}
is the Galtchouk-Kunita-Watanabe decomposition of $A_T$, where $\int_0^{\cdot} \delta^A_u \ud \hat S_u$ is $P$-strongly orthogonal to $L^A$, then $\delta^A$ is the unique benchmarked risk-minimizing strategy for $A$, i.e. the trading strategy which minimizes the expected quadratic risk as in Definition A.4 of \cite{Bia-Zha}. Furthermore, the associated benchmarked cumulative cost process $C_t^{\bar \delta}$ is given by 
\[
	C_t^{\bar \delta} = \esp{A_T} + L^A_t, \ \ \ \ t \in [0,T],
\]
and the benchmarked value process $\hat S_t^{\bar \delta}$ is given by the discounted value of the real-world pricing formula (\ref{eq: pricing formula dividend})
\[
	\hat S_t^{\bar \delta} = \condespg{A_T - A_t} =  \frac{I_t}{S^*_t} V_t, \ \ \ \ t \in [0,T].
\]
Note that decomposition (\ref{eq: GKW decomposition}) shows the orthogonality between the perfectly hedgeable part $\int_0^{T} \left( \delta^A_u \right)^\top \ud \hat S_u$ of $A_T$ and the totally unhedgeable part $\left( \esp{A_T} + L^A_T \right)$, covered by the benchmarked cumulative cost process $C$.

\subsection{Pricing and hedging non-life insurance claims}\label{sec: pricing non-life}
In the setting outlined above, we now apply the results of Section \ref{sec: preliminary results} to compute the real world pricing formula for non-life insurance claims, under the interpretation of Section \ref{sec: non life}. The cumulative payment at time $t$ related to $i$-th policy expressed in real value is given by
\[
	\sum_{j = 1}^{\infty} \Ind{\tau^i_j \leqslant t} X^i_j = \sum_{j = 1}^{\mathbf{N}^i_t} X^i_j.
\]
The nominal benchmarked cumulative payment process $A := (A_t)_{t \in [0,T]}$ is hence
\begin{equation}\label{eq: reserve}
	{A}_t := \int_0^t \frac{I_s}{S^*_s} \ud D_s = \sum^n_{i=1} \sum_{j = 1}^{\mathbf{N}^i_t} \frac{I_{\tau^i_j}}{S^*_{\tau^i_j}} X^i_j, \ \ \ \ t \in [0,T].
\end{equation}
Following \cite{Ary}, \emph{claim reserving problem} in the context of non-life insurance can be formulated as the estimation of $A$. In particular, we are not only interested in estimating the overall reserve $E[A_T]$, but also the conditional reserve $\condespg{A_T - A_t}$ at a given time $t$, i.e. a predictor of the remaining nominal payment $A_T - A_t$, given the information up to time $t$. Unlike the life insurance case, the risk related to non-life insurance policies is not only related to the accident itself, but also to the first reporting delay (this is the case of incurred but not reported claims, called IBNR claims), to the time and the size of developments after the first reporting.
We now focus on pricing and hedging the remaining nominal payment ${A}_T - A_t$, for $t \in [0,T]$.
We assume that the process $I/S^*$ is $\Fb$-conditionally independent of $\tau^i_1$, for all $i=1,...,n$, and that the cumulative payments related to marked point processes $(\tilde{\tau}^i, \tilde{X}^i_j)_{j \in \N_0}$, $i = 1,...,n$, 
\[
	\sum_{j = 1}^{\tilde{\mathbf{N}}^i_t} \tilde{X}^i_j, \ \ \ \ t \in [0,T], \ \ \ \ i = 1,...,n,
\]
are i.i.d. compound Poisson processes, i.e. $\tilde{\mathbf{N}}^i$ are  mutually independent Poisson processes with parameter $\lambda$, and $\tilde{X}^i_j$ are i.i.d. integrable nonnegative random variables independent of $\tilde{\mathbf{N}}^i$ with expectation $E[\tilde{X}^i_j] = m$. In this case, we have 
\[
	\tilde{m}(t)  = \lambda  m t,  \ \ \ \ t \in [0,T],
\]
where $\tilde m$ is defined in (\ref{eq: tilde m}).

In view of the above assumptions, all conditions in Theorem \ref{theo: G representation 2} are satisfied in the case of $Y = A_T - A_t$, for $t \in [0,T]$.
Let $R_t$ be the number of reported claims at time $t$, i.e. 
\[
	R_t := \sum_1^n \Ind{\tau^i_1 \leqslant t}, \ \ \ \ t \in [0,T].
\]
The real world pricing formula (\ref{eq: pricing formula dividend}) together with Corollary \ref{cor:reduction G to H}, Theorem \ref{theo: G representation 2} and  Assumption \ref{ass: concrete model} yield
\begin{align}
	V_t  \frac{I_t}{S^*_t} =&  \condespg{A_T - A_t} =   \condespg{\sum^n_{i=1} \sum_{j = \mathbf{N}^i_t}^{\mathbf{N}^i_T} \frac{I_{\tau^i_j}}{S^*_{\tau^i_j}} X^i_j} \nonumber\\
	=& \sum^n_{i=1}  E \left. \left[\sum_{j = \mathbf{N}^i_t}^{\mathbf{N}^i_T} \frac{I_{\tau^i_j}}{S^*_{\tau^i_j}} X^i_j\right|  \F_t \vee \H^i_t \right] \nonumber\\
	=& \lambda  m R_t \left. E \left[ \int_t^T   \frac{I_u}{S^*_u} \ud u  \right| {\H}^{i,1}_\infty \vee \F_t \right]  \nonumber\\
		& + (n - R_t) \frac{ E\left. \left[ \int_t^T \left(E[X^i_1]  \frac{I_u}{S^*_u} + \lambda m \int_u^T \frac{I_v}{S^*_v}  \ud v \right) \ud P\left. \left( {\tau^i_1 \leqslant u} \right| \F_{u} \right) \right| \F_t \right]}{e^{-\int^t_0 \mu_u \ud u} +  \int_0^t \bar{G}(t - u) e^{-\int^u_0 \mu_v \ud v} \mu_u \ud u} \nonumber\\
	=& \lambda  m R_t \int_t^T   \left. E \left[  \frac{I_u}{S^*_u}  \right| \F_t \right]  \ud u \nonumber\\
		& + (n - R_t) \frac{ E\left. \left[ \int_t^T \left(E[X^i_1]  \frac{I_u}{S^*_u} + \lambda m \int_u^T \frac{I_v}{S^*_v}  \ud v \right) \ud P\left. \left( {\tau^i_1 \leqslant u} \right| \F_{u} \right) \right| \F_t \right]}{e^{-\int^t_0 \mu_u \ud u} +  \int_0^t \bar{G}(t - u) e^{-\int^u_0 \mu_v \ud v} \mu_u \ud u} \nonumber\\
	=& \lambda  m R_t (T-t) \frac{I_t}{S^*_t}   \nonumber\\
		& + (n - R_t) \frac{ E\left. \left[ \int_t^T \left(E[X^i_1]  \frac{I_u}{S^*_u} + \lambda m \int_u^T \frac{I_v}{S^*_v}  \ud v \right) \ud P\left. \left( {\tau^i_1 \leqslant u} \right| \F_{u} \right) \right| \F_t \right]}{e^{-\int^t_0 \mu_u \ud u} +  \int_0^t \bar{G}(t - u) e^{-\int^u_0 \mu_v \ud v} \mu_u \ud u}, \label{eq: non life two components}
\end{align}
	where the conditional probability function $P\left. \left( {\tau^i_1 \leqslant t} \right| \F_{t} \right)$ is given in (\ref{eq: diff P tau 1}), i.e.
	\[
		P\left. \left( \tau^i_1 \leqslant t \right| \F_t \right) = \int_0^t \left( \alpha_0 e^{-\int^s_0 \mu_v \ud v} \mu_s +  \int_0^s g(s - u) e^{-\int^u_0 \mu_v \ud v} \mu_u \ud u \right)\ud s.
	\]

\noindent The first component on the left-hand side of (\ref{eq: non life two components}) 
\begin{equation}\label{eq: non life reported}
	\lambda  m R_t (T-t) \frac{I_t}{S^*_t} 
\end{equation}
corresponds to already reported claims. We observe that the valuation of this part does not involve any more the updating information after the first reporting.
The second component on the right-hand side of (\ref{eq: non life two components}) 
\begin{equation}\label{eq: non life not reported}
	(n - R_t)\frac{ E\left. \left[ \int_t^T \left(E[X^i_1]  \frac{I_u}{S^*_u} + \lambda m \int_u^T \frac{I_v}{S^*_v}  \ud v \right) \ud P\left. \left( {\tau^i_1 \leqslant u} \right| \F_{u} \right) \right| \F_t \right]}{e^{-\int^t_0 \mu_u \ud u} +  \int_0^t \bar{G}(t - u) e^{-\int^u_0 \mu_v \ud v} \mu_u \ud u},
\end{equation}
which can be further explicitly computed, corresponds to not reported claims and includes both cases of incurred but not reported (IBNR) claims as well as not yet incurred claims. The standard literature of non-life insurance is mainly focused on IBNR claims. However, for the pricing problem it is more appropriate to consider the entire expression (\ref{eq: non life two components}). 
As already mentioned at the beginning of this section, this price equals the benchmarked risk-minimizing price, if we assume square integrability of the claim. 

We now calculate the associated benchmarked risk-minimizing strategy. For additional details, we refer to \cite{Zhang}.
We mainly focus on the second component (\ref{eq: non life not reported}) related to not yet reported claims, since hedging strategy is additive with respect to claims, and the first component (\ref{eq: non life reported}) related to already reported claims is perfectly hedgeable by trading inflation linked zero-coupon bonds.
Using the same arguments of Proposition 4.11 in \cite{Bar} and Section 4.1 of \cite{Bia-Zha}, the benchmarked risk minimizing strategy $\bar{\delta}$ associated to not yet reported claims is given by
\[
	\bar{\delta}_t = (n - R_{t_{-}}) \left(e^{-\int^t_0 \mu_u \ud u} +  \int_0^t \bar{G}(t - u) e^{-\int^u_0 \mu_v \ud v} \mu_u \ud u \right)^{-1}\phi_t, \ \ \ \ t \in [0,T],
\] 
where $\phi_t$ is the benchmarked risk-minimizing strategy at $t$ related to 
\begin{equation}\label{eq: non life fin claim}
	U_t := E\left. \left[ \int_0^T \left(E[X^i_1]  \frac{I_u}{S^*_u} + \lambda m \int_u^T \frac{I_v}{S^*_v}  \ud v \right) \ud P\left. \left( {\tau^i_1 \leqslant u} \right| \F_{u} \right) \right| \F_t \right], \ \ \ \ t \in [0,T].
\end{equation}
More precisely, the vector process $\phi := (\phi_t)_{t\in[0,T]}$ follows from the Galtchouk--Kunita--Watanabe decomposition of $(U_t)_{t \in [0,T]}$,
\begin{align*}
	U_t =& E\left[ \int_0^T \left(E[X^i_1]  \frac{I_u}{S^*_u} + \lambda m \int_u^T \frac{I_v}{S^*_v}  \ud v \right) \ud P\left. \left( {\tau^i_1 \leqslant u} \right| \F_{u} \right) \right] + \int_0^t \phi^{\top}_u \ud \hat{S}_u + L^U_t.
\end{align*}
Note that benchmarked risk-minimizing strategy covers only the perfectly hedgeable part of the liability cash flow $A$. It is however the best possible hedging strategy in the benchmarked risk-minimizing sense. Indeed, the totally unhedgeable part, which is orthogonal to the hedgeable part, represents the basis hedging risk.

The form of $V$ in (\ref{eq: non life two components}) suggests the design of derivatives which can be used to hedge risks in this market model.
In particular, since $U$ involves only the stochastic processes  $S^*$, $I$ and $\mu$, it is sufficient to introduce three kinds of instruments, namely pure financial assets, inflation linked derivatives and macro risk-factor linked derivatives, to hedge risks derived from $S^*$, $I$ and $\mu$ respectively.

\section{Conclusion}

In this paper, we introduce a general framework for modeling an insurance liability cash flow in continuous time by extending the reduced-form setting. This framework allows to consider a nontrivial dependence between the reference information flow and the internal insurance information flow. 
In this setting, we compute explicit valuation and hedging formulas, which can be used for pricing non-life insurance products under the benchmark approach.

\section*{Acknowledgements}

The authors would like to thank Irene Schreiber for interesting discussions about property and casualty insurance and anonymous referees for helpful remarks.

\bibliographystyle{plain}
\bibliography{bibliography}

\end{document}